\documentclass[draft]{article}
\usepackage[a4paper,margin=25mm]{geometry}
\usepackage{amsmath,amssymb,tikz}

\parskip\smallskipamount

\let\set\mathbb
\def\<#1>{\langle#1\rangle}

\usepackage{amsthm,amsmath}
\usepackage{tikz}
\usepackage{xfrac}
\usepackage{cite}

\usepackage{enumitem}
\usepackage{mathtools}

\newtheorem*{lemma*}{Lemma}
\newtheorem{theorem}{Theorem}[section]
\newtheorem{prop}[theorem]{Proposition}
\newtheorem{corollary}[theorem]{Corollary}
\newtheorem{lemma}[theorem]{Lemma}
\newtheorem{remark}[theorem]{Remark}
\newtheorem{algorithm}[theorem]{Algorithm}

\newtheorem{definition}[theorem]{Definition}
\newtheorem{example}[theorem]{Example}

\newtheorem{convention}[theorem]{Convention}

\makeatletter
\newcommand{\myitem}[1]{%
	\item[(#1)]\protected@edef\@currentlabel{#1}%
}
\makeatother

\def\eatspace#1{#1}
\def\step#1#2{\par\kern1pt\hangindent#2em\hangafter=1\noindent\rlap{\small#1}\kern#2em\relax\eatspace}

\let\set\mathbb
\def\<#1>{\langle#1\rangle}

\newcommand{\va} { {\bf a}}

\newcommand{\vc} { {\bf c}}

\newcommand{\bN} { {\mathbb{N}}}

\newcommand{\bC} { {\mathbb{C}}}
\newcommand{\bQ} { {\mathbb{Q}}}
\newcommand{\bZ} { {\mathbb{Z}}}

\newcommand{\Z} { {\mathcal{Z}}}

\def\Span{\operatorname{Span}}
\def\coeff{\operatorname{Coeff}}
\def\row{\operatorname{Row}}
\def\qsquo{\operatorname{qsquo}}
\def\lcm{\operatorname{lcm}}

\def\ind{\operatorname{ind}}
\def\Sing{\operatorname{Sing}}
\def\Extra{\operatorname{Extra}}

\def\Gal{\operatorname{Gal}}

\def\ord{\operatorname{ord}}

\def\lclm{\operatorname{lclm}}

\def\lclm{\operatorname{lclm}}

\begin{document}

\title{Symmetric Division of Linear Ordinary Differential Operators
	\thanks{L.\ Du was supported by the Austrian FWF grant 10.55776/I6130. M.\ Kauers was supported by the Austrian FWF grant 10.55776/PAT8258123, 10.55776/I6130, and 10.55776/PAT9952223.
	}}

\author{Lixin Du and Manuel Kauers\\
	\normalsize Institute for Algebra, Johannes Kepler University, Linz, A4040, Austria\\
	{\sf \normalsize  lixindumath@gmail.com, manuel.kauers@jku.at}
}
\date{}
\maketitle
\begin{abstract}
  The symmetric product of two ordinary linear differential operators $L_1,L_2$ is
  an operator whose solution set contains the product $f_1f_2$ of any solution
  $f_1$ of $L_1$ and any solution $f_2$ of~$L_2$. It is well known how to compute
  the symmetric product of two given operators $L_1,L_2$. In this paper we consider
  the corresponding division problem: given a symmetric product $L$ and one of its
  factors, what can we say about the other factors?
\end{abstract}
\section{Introduction}

This work is about linear differential operators with rational function coefficients,
i.e., operators that can be written in the form 
\[
  p_0 + p_1 D + \cdots + p_r D^r,
\]
where $D$ refers to the derivation with respect to~$x$, and $p_0,\dots,p_r$ are
certain rational functions in~$x$. Such operators act in a natural way on differential
fields, for example on the field of formal Laurent series. The result of applying the
above operators to a series $f$ is meant to be the series
\[
  p_0(x)f(x) + p_1(x)f'(x) + \cdots + p_r(x)f^{(r)}(x),
\]
where by $p_i(x)$ we mean the series expansions of the rational function~$p_i$.

If $L$ is an operator and $f$ is a series, we write $L\cdot f$ for the series
resulting from applying $L$ to~$f$. A series $f$ is called
\emph{D-finite}~\cite{stanley80,kauers23} if there exists a nonzero operator $L$ such that
$L\cdot f=0$. Such an $L$ is called an \emph{annihilating operator}
for~$f$. D-finite series play an important role in computer algebra. There are
many algorithms for solving problems about D-finite series, and these algorithms
nowadays are routinely applied in areas in which such problems naturally arise.

A D-finite series is uniquely determined by an annihilating operator and a
finite number of initial terms. For this reason, algorithms for D-finite series
rely heavily on computations with operators. To enable computations with
operators, the set $C(x)[D]$ of all operators is turned into a ring by defining
addition and multiplication in such a way that the action of this ring on the
field $C((x))$ of Laurent series via operator application turns that field into
a $C(x)[D]$-module. This means that addition and multiplication are set up in
such a way that we have $(L+M)\cdot f=L\cdot f+M\cdot f$ and $(LM)\cdot
f=L\cdot(M\cdot f)$ for every $L,M\in C(x)[D]$ and every $f\in C((x))$. The
resulting ring $C(x)[D]$ of differential operators is an example of an Ore
algebra~\cite{ore33,BronsteinPetkovsek1996,kauers23}. Its multiplication is not commutative but governed by the
commutation rule $Dx=xD+1$, which reflects the product rule for differentiation.

The arithmetic in the ring $C(x)[D]$ of operators is thus quite different from the
arithmetic in the field $C((x))$. In particular, if $L$ and $M$ are annihilating operators
of $f$ and~$g$, respectively, then $L+M$ is usually not an annihilating operator of $f+g$,
and $LM$ is usually not an annihilating operator of $fg$. Nevertheless, if $f$ and $g$ are
D-finite, then so are their sum $f+g$ and their product~$fg$.
An annihilating operator for $f+g$ can be obtained from $L$ and $M$ by taking a
common left multiple of these operators, i.e., an operator that can be written as $AL$
and also as $BM$ for certain operators $A,B$. Such operators always exist, and there is
one of minimal order which is unique up to left-multiplication by nonzero rational functions.
This operator is called the \emph{least common left multiple} of $L$ and~$M$.
See \cite{BostanChyzakSalvyLi2012lclm,kauers14f,kauers23} for information about the
computation of such common left multiples.

Similarly, there is a construction by which an annihilating operator for the product $fg$
can be obtained from $L$ and~$M$. Again, among all these operators there is one of minimal
order, and this one is unique up to left-multiplication by nonzero rational functions.
It is called the \emph{symmetric product} of $L$ and $M$ and denoted by $L\otimes M$. 
This product is not to be confused with the product $LM$ obtained via the multiplication
in the ring $C(x)[D]$. See \cite{BronsteinMuldersWeil1997,gaillard25,kauers23} for more
about the computation of symmetric products.

So we have two distinct kinds of multiplication for operators: the regular
product and the symmetric product. What are the corresponding divisions? For the
regular product, this is easy to answer. With respect to this product, despite
the lack of commutativity the ring $C(x)[D]$ very much behaves like a
commutative univariate polynomial ring. In particular, it is a left Euclidean
domain~\cite{ore33,BronsteinPetkovsek1996,kauers23}. We have a notion of
division with remainder which works very much like ordinary polynomial division,
and we have a Euclidean algorithm. In fact, the extended Euclidean algorithm in
$C(x)[D]$ gives rise to one way of computing least common left multiples of
operators.

It is less clear how to do division with respect to the symmetric
product. Apparently, this question has not been systematically addressed before,
and the purpose of the present article is to develop some theory and algorithms
for this division. The task under consideration is, given two operators $M$
and~$L$, to find another operator, $Q$, such that $M=L\otimes Q$. We call such
an operator $Q$ a \emph{symmetric quotient} of $M$ and~$L$. The solutions $g$
of a symmetric quotient have the property that for every solution $f$ of~$L$,
the product $fg$ is a solution of~$M$. Note that this is not the same as trying
to compute an annihilating operator for the quotient $f/g$ of a solution $f$ of
$L$ by a solution $g$ of~$M$. Indeed, these quotients are usually not D-finite~\cite{harris85}.

We were led to study symmetric division by an attempt to construct a cryptographic
key exchange system based on operator arithmetic. The idea was that Alice chooses
two operators $L$ and $A$ and sends $L$ and $L\otimes A$ to Bob. Bob chooses an
operator $B$ and sends $L\otimes B$ back to Alice. Knowing $L\otimes B$ and~$A$,
Alice can compute $L\otimes A\otimes B$ (note that the symmetric product is
commutative, unlike the regular operator product), and knowing $L\otimes A$ and~$B$,
Bob can compute $L\otimes A\otimes B$ as well, so they have constructed a shared secret. 

The rationale of this crypto system was that while the symmetric product of two
operators can be efficiently computed, it is unclear how to do symmetric
division, so an attacker won't easily be able to recover $A$ from the knowledge
of both $L$ and $L\otimes A$. In a way, our main result is that this crypto
system is not secure, because symmetric division can be done. Although we cannot
solve the problem in full generality, our algorithms suffice to render the idea
obsolete.

While this motivation may seem a bit far fetched, we believe that symmetric division
is of interest in its own right and that the ideas behind our algorithms are worth
being shared. A key tool is the concept of colon spaces, an adaption of the definition
of colon ideals to vector space, which we introduce in Section~\ref{sec:colonspace} after reviewing
in Section~\ref{sec:prelim} the relevant background for this paper. In Section~\ref{sec:divisionalg} we
present an algorithm for computing what we will call local quasi-symmetric quotients.
This is a variant of the symmetric division problem that we found most tractable.
The algorithm of Section~\ref{sec:divisionalg} depends on a number of subroutines which are
detailed in Sections \ref{sec:truncation}--\ref{sec:degree}. 
In Section~\ref{sec:special} we discuss how (global) quasi-symmetric quotients can be constructed
for certain special kinds of operators.

\section{Preliminaries}\label{sec:prelim}

Throughout this paper, let $C$ be an algebraically closed field of
characteristic zero. Let $C(x)$ be the field of rational functions in $x$ over $C$. Let $C[[x]]$ be the ring of formal power series and let $C((x))$ be its quotient field, i.e., the field of formal Laurent series.

\subsection{Truncated Series}

For any $k\in \bZ$, let $T_{k}: C((x)) \to C((x))$ be the map defined by
\[T_k(f) = \sum_{i=j}^k a_i x^i,\]
for all $f=\sum_{i=j}^\infty a_ix^i\in C((x))$ with $a_i\in C$. The expression $T_k(f)$ is called a \emph{truncation of $f$ at precision~$k$}. We may use the notation
\[T_k(f)= \sum_{i=j}^k a_i x^i + O(x^{k+1})\]
to make the truncation precision more explicit. We recall some basic properties of truncated series,  which follow directly from the definitions of addition and multiplication of series, see~\cite[\S 1.1]{kauers23} for details.
\begin{lemma}\label{lem: series add}
	Let $f=\sum_{i= \lambda_0}^{\lambda}a_ix^i+O(x^{\lambda+1})$ and $g=\sum_{i= \mu_0}^{\mu}b_ix^i+O(x^{\mu+1})$ be Laurent series in $C((x))$, where $\lambda_0,\lambda,\mu_0,\mu\in\bZ$ with $\lambda\geq \lambda_0$, $\mu\geq \mu_0$. Then
	\[f+g=\sum_{i=\min\{\lambda_0,\mu_0\}}^{\min\{\lambda,\mu\}} (a_i+b_i)x^i +O(x^{\min\{\lambda,\mu\}+1}).\]
\end{lemma}

\begin{lemma}\label{lem: series multi}
	Let $f=x^{\lambda_0}\sum_{i= 0}^{\lambda_1}a_ix^i+O(x^{\lambda_0+\lambda_1+1})$ and $g=x^{\mu_0}\sum_{i=0}^{\mu_1}b_ix^i+O(x^{\mu_0+\mu_1+1})$ be Laurent series in $C((x))$, where $\lambda_0,\mu_0\in\bZ$ and $\lambda_1,\mu_1\in\bN$. 
	\begin{enumerate}[label=(\roman*)]
		\item\label{it:series1}  The product $fg$ satisfies	\[fg=x^{\lambda_0+\mu_0}\sum_{i=0}^{\min\{\lambda_1,\mu_1\}} \left(\sum_{j=0}^ia_jb_{i-j}\right)x^i +O(x^{\min\{\lambda_1,\mu_1\}+\lambda_0+\mu_0+1}).\]
		\item\label{it:series2} If $a_{0}\neq 0$, then $f$ is invertible and 
		\[f^{-1}= x^{-\lambda_0}\left(\sum_{i=0}^{\lambda_1}c_i x^i\right) + O(x^{\lambda_1-\lambda_0+1}),\]
		where $c_0=\frac{1}{a_0}$, $c_i=-\frac{1}{a_0}(\sum_{j=1}^i a_jc_{i-j})$ for all $1\leq i\leq \lambda_1$.
	\end{enumerate}
\end{lemma}
\begin{corollary}\label{cor: series division}
	Let $r\in \bN\setminus\{0\}$ and let $f_i=\sum_{j=i-1}^{\infty}a_jx^j, g_i=\sum_{j=i-1}^\infty b_jx^j$ be power series in $C[[x]]$, where $1\leq i\leq r$ and $b_{i-1}\neq 0$ (while $a_{i-1}$ may be zero). Then for all $k\geq 0$, we have
	\[T_k\left(\frac{f_i}{g_i}\right) = T_k\left(\frac{T_{k+r-1}(f_i)}{T_{k+r-1}(g_i)}\right).\]
\end{corollary}
\begin{proof}
	Since $T_{k+r-1}(g_i)=\sum_{j=i-1}^{k+r-1}b_jx^j +O(x^{k+r}) = x^{i-1}\sum_{j=0}^{k+r-i}b_{j+i-1}x^j + O(x^{k+r})$, by Lemma~\ref{lem: series multi}.\ref{it:series2} we have $\frac{1}{g_i}= \frac{1}{T_{k+r-1}(g_i)} + O(x^{k+r-2i+2})$ for all $1\leq i\leq r$. Moreover, $\frac{1}{T_{k+r-1}(g_i)}= x^{-(i-1)}\sum_{j=0}^\infty c_j x^j$ for some $c_j\in C$. Since $f_i= T_{k+r-1}(f_i) + O(x^{k+r})$, by Lemma~\ref{lem: series multi}.\ref{it:series1} we get \[\frac{f_i}{g_i} =\frac{T_{k+r-1}(f_i)}{T_{k+r-1}(g_i)} + O(x^{\min\{k+r-1-(i-1), \,k+r-2i+1+(i-1)\}+(i-1)-(i-1)+1}).\]
	Therefore, $\frac{f_i}{g_i}-\frac{T_{k+r-1}(f_i)}{T_{k+r-1}(g_i)} = O(x^{k+r-i+1}) = O(x^{k+1})$ for all $1\leq i\leq r$. This completes the proof.
\end{proof}

\subsection{The ring of linear differential operators}

Let $C(x)[D]$ be an Ore algebra, where $D$ is the differentiation with respect to $x$ and satisfies the commutation rule $Dx=xD+1$. Operators in $C(x)[D]$ have the form $L=\ell_0+\ell_1 D+\cdots +\ell_rD^r \in C(x)[D]$ with $\ell_0,\ell_1,\ldots,\ell_r\in C(x)$. When $\ell_r\neq 0$, we call $\ord(L) := r$ the \emph{order} of~$L$. Let $F$ be a differential ring and write $'$ for its derivation. The Ore algebra $C(x)[D]$ acts on $F$ via
\[(\ell_0+\ell_1D+\cdots+\ell_rD^r)\cdot f=\ell_0 f + \ell_1f' + \cdots + \ell_rf^{(r)}.\]
An element $f\in F$ is called a \emph{solution} of an operator $L\in C(x)[D]$ if $L\cdot f =0$. For $L \in C(x)[D]$, we call $V_F(L) := \{f\in F\mid L\cdot f =0\}$ the \emph{solution space} of $L$ in $F$.  For convenience, we write $V(L)$ to denote the solution space of $L$ in the field of formal Laurent series $C((x))$. An element $c\in F$ is called a \emph{constant} if $D \cdot c = 0$. The set of constants in a ring forms a subring and in a field forms a subfield.



The Ore algebra $C(x)[D]$ is a right Euclidean domain, and the (extended) Euclidean algorithm carries over almost literally to this setting~\cite{BronsteinPetkovsek1996,kauers23}. For every $L_1,L_2\in C(x)[D]$, $L_1\neq 0$, there exist unique $Q,R\in C(x)[D]$ such that $L_2 = QL_1 +R$ and $\ord(R)<\ord(L_1)$. If $R=0$, we say that $L_1$ is a \emph{right factor} of $L_2$ and that $L_2$ is a \emph{left multiple} of $L_1$. 
An element $f\in F$ is called \emph{D-finite} if there exists a nonzero operator $L\in C(x)[D]$ such that $L\cdot f= 0$. Such an $L$ is called an \emph{annihilator} of $f$. Among all  annihilators, one of minimal order is called a \emph{minimal annihilator}. Since every left ideal of $C(x)[D]$ is left principal, every annihilator of $f$ is a left multiple of its minimal annihilator. The following lemma describes properties of the solution spaces of an operator and its right factors.
\begin{lemma}[\!\!{\cite[Lemma 2.1]{Singer1996reducibility}}]\label{lem:sol-rightfactor}
	Let $L_1,L_2\in C(x)[D]$ and assume that $\ord(L_1)= r_1$, $\ord(L_2)= r_2$. Let $F$ be a differential field extension of $C(x)$ having the same constant $C$. 
	\begin{enumerate}[label=(\roman*)]
		\item\label{it:rf1} $\dim_CV_F(L_1) \leq r_1$.
		\item\label{it:rf2}  If $\dim_CV_F(L_1)=r_1$ and $V_F(L_1)\subseteq V_F(L_2)$, then $L_1$ is a right factor of~$L_2$.
		\item\label{it:rf3} If $\dim_CV_F(L_1) = r_1$ and $L_2$ is a right factor of $L_1$, then $\dim_CV_F(L_2) = r_2$ and
		$V_F(L_2) \subseteq V_F(L_1)$.
	\end{enumerate}
\end{lemma}

An operator $L\in C(x)[D]$ is called a \emph{common left multiple} of two operators $L_1,L_2 \in C(x)[D]$ if there exist $R_1,R_2 \in C(x)[D]$ such that $L=R_1L_1= R_2L_2$. Among all such common left multiples, one of minimal order is called a \emph{least common left multiple~(lclm)}. We write $\lclm(L_1,L_2)$ for the unique primitive lclm of $L_1$ and $L_2$, i.e., the lclm whose coefficients are coprime polynomials in $C[x]$ and whose leading coefficient is primitive in $x$. A key feature of the lclm is that whenever $f_1$ is a solution of $L_1$ and $f_2$ is a solution of $L_2$, then their sum $f_1+f_2$ is a solution of $\lclm(L_1,L_2)$.  For the efficient computation of the least common left multiple, see~\cite{BostanChyzakSalvyLi2012lclm}. There is a similar construction for multiplication. For any two nonzero differential operators $L_1,L_2\in C(x)[D]$, there exists a unique primitive operator $L_1\otimes L_2$ of lowest order, called the \emph{symmetric product} of $L_1$ and~$L_2$, such that
whenever $f_1$ is a solution $L_1$ and $f_2$ is a solution of $L_2$, then their product $f_1f_2$ is a solution of $L_1\otimes L_2$. As a special case, the \emph{$s$-th symmetric power} of an operator $L\in C(x)[D]$ is defined as $L^{\otimes s} =L\otimes \cdots\otimes L$. For the efficient computation of the symmetric powers, see~\cite{BronsteinMuldersWeil1997}. Note that unlike the multiplication in~$C(x)[D]$, the symmetric product is commutative. We recall the following properties of the lclm (see~\cite[Lemma 3.2]{Singer93} or \cite[\S 4.2, Ex.~13; solution on p.~578]{kauers23}) and of the symmetric product (see \cite[Corollary 2.9]{VanderputSinger2003Galois} or~\cite[Proposition 2.6]{Singer1979algebraic}).
\begin{lemma}\label{lem:sol-lcml-sym}
	Let $L_1,L_2\in C(x)[D]$. Let $F$ be a differential field extension of $C(x)$ having the same constant $C$. Assume that $\dim_{C}(V_F(L_1)) = \ord(L_1)$ and  $\dim_{C}(V_F(L_2)) = \ord(L_2)$. Then
	\begin{enumerate}[label=(\roman*)]
		\item\label{it:sol-lclm} $\dim_{C}V_F(\lclm(L_1,L_2)) = \ord(\lclm(L_1,L_2))$ and
		\begin{equation*}
			V_F(\lclm(L_1,L_2)) = V_F(L_1) + V_F(L_2);
		\end{equation*}
		
		\item\label{it:sol-sym} $\dim_{C}V_F(L_1\otimes L_2) = \ord(L_1\otimes L_2)$ and
		\begin{equation*}
			V_F(L_1 \otimes L_2) = \Span_C\{gh \mid g \in V_F(L_1),\, h \in V_F(L_2)\}.
		\end{equation*}
	\end{enumerate}
\end{lemma}

The following lemma (see \cite[\S 4.1 Ex.~23; solution on p.~575]{kauers23}) shows that the symmetric product is distributive over the lclm.
\begin{lemma}\label{lem:distributivity}
	Let $L, Q_1, Q_2\in C(x)[D]$. Then $L\otimes (\lclm(Q_1,Q_2)) =\lclm(L\otimes Q_1, L\otimes Q_2)$.
\end{lemma}
As a consequence of Lemma~\ref{lem:distributivity}, we obtain the following corollary.
\begin{corollary}\label{cor:lclm-quotient-2}
	Let $L,Q_1,Q_2,M\in C(x)[D]$. If both $L\otimes Q_1$ and $ L\otimes Q_2$ are right factors of $M$, then $L\otimes(\lclm(Q_1,Q_2))$ is also a right factor of $M$. In particular, if $M=L\otimes Q_1 =L\otimes Q_2$, then $M=L\otimes (\lclm(Q_1,Q_2))$.
\end{corollary}

\subsection{Indicial polynomials}
We recall properties of linear differential equations in~\cite[Chap. XVI, XVII]{Ince1926} or \cite[Chap. V]{Poole1936}. Let $N\in C[x][D]$ be a linear differential operator with polynomial coefficients. If $N$ is given with rational coefficients, its denominators can be cleared. Then the action of $N$ on a monomial $z^s$ with $s\in\bN$ yields a polynomial
\begin{equation}\label{eq:Nx^s}
	N\cdot x^s = x^{s+\sigma_{N}}(p_0(s)+p_1(s)x+\cdots+p_t(s)x^t),
\end{equation}
where $\sigma_N\in \bZ$, $t\in \bN$, $p_i\in C[s]$ and $p_0\neq 0$. The coefficients of $p_i(s)$ depend on those coefficients of~$N$. The polynomial $p_0(s)$ is called the \emph{indicial polynomial} of $N$ at $0$, denoted by $\ind_0(L)$. By linear combination, the action of $N$ on a formal power series $f=\sum_{i=0}^\infty c_ix^i\in C[[x]]$ is the formal power series
\[N\cdot f = \sum_{i=0}^\infty c_i (N\cdot x^i)=\sum_{k=0}^\infty \left(c_0 p_k(0) + \cdots + c_kp_0(k)\right)x^{k+\sigma_N},\]
where $p_i(x)=0$ if $i>t$. Then the differential equation $N\cdot f= 0$ implies that
\begin{equation}\label{eq:Nf-rec1}
	c_0p_0(0) = 0,\, c_1p_0(1)+c_0p_1(0)=0,\,\ldots,\, c_{t-1}p_0(t-1)+\cdots+c_0p_{t-1}(0)=0
\end{equation}
and the linear recurrence of order $t$
\begin{equation}\label{eq:Nf-rec2}
	c_kp_0(k)+\cdots + c_{k-t}p_t(0)  =0,\quad \forall\,k\geq t.
\end{equation}
Let 
\begin{equation}\label{eq:Z_N}
	\Z_N := \{k\in \bN\mid p_0(k)=0\}
\end{equation}
be the set of nonnegative integer roots of the indicial polynomial of $N$ at $0$. For all $k\notin \Z_N$, the coefficient $c_k$ is determined from the previous ones. In particular, this discussion leads to the following basic property of power series solutions, see~\cite{Ince1926}.  
\begin{lemma}\label{lem:series-sol-rec}
	Let $N\in C[x][D]$ be such that $N\cdot x^s=x^{s+\sigma_N}(p_0(s)+p_1(s)x+\cdots + p_t(s)x^t)$ is in the form~\eqref{eq:Nx^s}. Let $k_0=\max\Z_N$. If $f=\sum_{i=0}^\infty c_ix^i\in C[[x]]$ is a formal power series solution of $N$, then,	for all $k> k_0$,
	\[c_k =-\frac{1}{p_0(k)}\left(\,\sum_{i=1}^{\min\{t,k\}}c_{k-i}p_i(k-i)\right).\]
\end{lemma}
\begin{proof}
	Since $p_0(k)\neq 0$ for all $k> k_0$, the result follows from the equations \eqref{eq:Nf-rec1} and~\eqref{eq:Nf-rec2}. 
\end{proof}

\subsection{Generalized series solutions}
Let $L=\ell_0 + \ell_1D+\cdots+\ell_rD^r\in C(x)[D]$ be a fixed operator of order $r$. A point $ \xi \in C$ is called a \emph{singularity} of $L$ if it is a pole of one of the rational functions $\ell_0/\ell_r, \ldots,\ell_{r-1}/\ell_r$. The point $\infty$ is called a \emph{singularity} if, after the substitution $x\mapsto x^{-1}$, the origin $0$ becomes a singularity. A point $\xi\in C \cup \{\infty\}$ which is not a singularity is called an \emph{ordinary point} of~$L$. If $\xi$ is an ordinary point of $L$, then after the change of variables $x\mapsto x+\xi$, zero becomes an ordinary point of the transformed operator. If $0$ is an ordinary point of $L$, then $L$ has $r$ linearly independent solutions in $C[[x]]$, see the following lemma. Since the dimension of the solution space $V(L) \subseteq C((x))$ can not exceed its order, these solutions also form a basis of $V(L)$.

%

\begin{lemma}[{\cite[Theorem 3.16]{kauers23}}]\label{lem:ordinary_fund_sys}
	Let $L\in C(x)[D]$ be an operator of order $r$. If $0$ is an ordinary point of~$L$, then $L$ has $r$ linearly independent solutions in $C[[x]]$ of the form
	\[g_1 = 1 + O(x^r),\quad g_2 = x + O(x^r),\quad\ldots,\quad g_r = x^{r-1} + O(x^r).\]
\end{lemma}

A singularity $\xi \in C\cup\{\infty\}$ of~$L$ is called \emph{apparent}  if the solution space of $L$ in $C[[x-\xi]]$ (or $C[[x^{-1}]]$ if $\xi=\infty$) has dimension $r$. A singularity $\xi\in C\cup\{\infty\}$ is called \emph{regular} if the indicial polynomial of $L$ at $\xi$ has degree $r$, and it is called \emph{irregular} otherwise. For each $\xi\in  C$, an operator $L$ of order $r$ admits $r$ linearly independent solutions of the form
\begin{equation}\label{EQ:generalized_sol_a}
	(x-\xi)^\alpha \exp(p((x-\xi)^{-1}))b(x-\xi, \log(x-\xi))
\end{equation}
for some $\alpha\in C$, $p\in C[x^{1/v}]$ and $b\in C[[x^{1/v}]][y]$ with $v\in\set N\setminus\{0\}$ and $p(0)=0$. Such objects are called \emph{generalized series solutions} at $\xi$, see~\cite{vanHoeij97b,kauers23}. For $\xi=\infty$, the operator $L$ admits $r$ linearly independent solutions of the form
\begin{equation}\label{EQ:generalized_sol_infty}
	x^{-\alpha}\exp(p(x))b(x^{-1}, \log(x))
\end{equation}
for some $\alpha\in C$, $p\in C[x^{1/v}]$and $b\in C[[x^{1/v}]][y]$ with $v\in\set N\setminus\{0\}$ and $p(0)=0$. If $\xi$ is a regular singularity, then all series solutions of $L$ at $\xi$ have $p=0$ and $v=1$. 
As a change of variables can always bring us back to the case $\xi=0$, it suffices to consider the case $x=0$. Let $C[[[x^{1/v}]]]$ be the ring of all finite $C$-linear combinations of series of the form $x^\alpha b(x,\log(x))$ with $\alpha\in C$ and $b\in C[[x^{1/v}]][y]$.

Since the series solutions in the form~\eqref{EQ:generalized_sol_a} (or \eqref{EQ:generalized_sol_infty} at $\infty$) may have fractional exponents, we consider $L\in C[x^{1/v}][D]$ in the following lemma. The indicial polynomial of an operator in $C[x^{1/v}][D]$ is defined similarly to the classical case, see~\cite[Definition 3.34]{kauers23}. 
\begin{definition}
	For a series $f\in x^\alpha C[[x^{1/v}]][\log(x)]$, a term $x^\beta \log(x)^\gamma$ is called an \emph{initial term} of~$f$ if $\beta$ is minimal among all exponents of $x$ appearing in $f$, and among the terms with exponent $\beta$, it has minimal $\gamma$. The exponent $\beta$ is called the \emph{local exponent} of $f$.
\end{definition} 

If $\alpha$ is a $\mu$-fold root of the indicial polynomial of $L$ at $0$, then $L$ has $\mu$ linearly independent solutions in $x^\alpha C[[x]][\log(x)]$ starting with the initial terms $x^\alpha \log(x)^{\gamma}$ for some~$\gamma$. More precisely, the following result can be obtained from further computations based on~\cite[Theorem 3.38 (item 2), Theorem 3.45]{kauers23}.
\begin{lemma}\label{lem: indicial - log}
	Let $L\in C[x^{1/v}][D]$ for some $v\in\bN \setminus\{0\}$. Suppose that the indicial polynomial $\ind_0(L)$ of $L$ at $0$ factorizes as  \[\ind_0(L) = c (s-\alpha_1)^{\mu_1}\cdots (s-\alpha_I)^{\mu_I},\] where $c\in C\setminus\{0\},\mu_1,\ldots, \mu_I\in \bN\setminus\{0\}$, the roots $\alpha_1,\ldots, \alpha_I\in C$ are distinct. Then the solution space of $L$ in $C[[[x^{1/v}]]]$ has a basis $g_{i,j}\, (i=1,\ldots,I,\, j = 1,\ldots,\mu_i)$ in the form:
	\begin{align*}
		g_{1,1}&= x^{\alpha_1} + \cdots, && \cdots, &g_{I,1}&= x^{\alpha_I} + \cdots, \\
		g_{1,2}&=x^{\alpha_1} \log(x) + \cdots, && \cdots, & g_{I,2}&= x^{\alpha_I} \log(x) + \cdots, \\
		& \vdots && && \vdots \\
	    g_{1,\mu_1}&=x^{\alpha_1} \log(x)^{\mu_1 - 1} + \cdots, && \cdots, & g_{I,\mu_I -1}&=x^{\alpha_I} \log(x)^{\mu_I-1} + \cdots.
	\end{align*}
	where $x^{\alpha_i} \log(x)^{j-1}$ is the initial term of $g_{i,j}$. 
\end{lemma}

The operator $L$ is called {\em Fuchsian} if all its singularities in $C\cup\{\infty\}$ are regular. Let $L\in C(x)[D]$ be a Fuchsian operator. For each $\xi\in C\cup\{\infty\}$, let
\begin{equation}\label{EQ: S_xi(L)-fuchsian}
	S_\xi (L):= \sum_{j=1}^r e_j(\xi) - \frac{r(r-1)}{2}
\end{equation}
where the numbers $e_j(\xi)$ are the local exponents of $L$ at $\xi$ (they are the roots of the indicial polynomial of $L$ at $\xi$). If $\xi$ is an ordinary point, then, by Lemma~\ref{lem:ordinary_fund_sys}, we have $S_\xi(L) =0$. 
The Fuchs relation (see~\cite[\S 15.4]{Ince1926} or~\cite[\S 20]{Poole1936}) states that
\begin{equation}\label{EQ: fuchs rel}
	\sum_{\xi\in \Sing(L) \cup \{\infty\}} S_\xi (L) = -r(r-1)
\end{equation}
where $\Sing(L)$ is the set of singularities of $L$ in $C$. 

\begin{example}\label{ex:2*3=4,Q}
	The operator $Q=(x-1)^3D^3 - 3(x-1)^2D^2 + 6(x-1)D - 6\in \bC(x)[D]$ is Fuchsian, with two regular singularities at $1$ and $\infty$. At the point $1$, the indicial polynomial \[\ind_1(Q) = (s-1)(s-2)(s-3)\] has degree 3, corresponding to the series solutions $(x - 1), (x - 1)^2, (x - 1)^3$ with local exponents $1$, $2$, $3$ respectively. At the point $\xi=\infty$, the indicial polynomial \[\ind_\infty(Q) = (s+3)(s+2)(s+1)\] also has degree 3, corresponding to the series solutions $x^3\,(1+O(\frac{1}{x})), x^2(1+O(\frac{1}{x})), x(1+O(\frac{1}{x}))$ with local exponents $-3$, $-2$, $-1$ respectively. Thus, for the operator $Q$:
	\begin{align*}
		S_1(Q) &= 1+2+3 - 3 = 3,\\
		S_\infty(Q) &= -3-2-1 -3 = -9,
	\end{align*}
	and the Fuchs relation~\eqref{EQ: fuchs rel} reduces to
	\[3 - 9 = - 6.\]
\end{example}

Let $L = \ell_0 + \ell_1 D +\cdots + \ell_r D^r \in C[x][D]$. For each term $x^j$ occurring in $\ell_i(x)$, draw a halfline in the plane that starts at $(i, j)$ and continues in the direction $(-1, -1)$, and determine the convex hull of all these halflines. The boundary of this convex hull is called the \emph{Newton polygon} of $L$ at $0$. Then every slope $1-c$ width $w$ corresponds to $w$ linearly independent solutions with an exponential part in the form $\exp(p(x^{-1}))$ with $\deg(p) =-c$ for some $p\in\bigcup_{v\in\bN\setminus\{0\}} C[x^{1/v}]$, see~\cite[\S 3.4]{kauers23}. 

Suppose that $f(x) = \exp(p(x^{-1})) g(x)$ is a solution of $L$ for some series $g(x)$, where $p\in C[x^{1/v}]$ with $v\in \bN\setminus\{0\}$. Then $g(x)$ is a solution of $\tilde L = \exp(-p(x^{-1}))\,L\,\exp(p(x^{-1}))$. 
\begin{definition}\label{Def: generalized indicial poly}
	Let $L\in C[x][D]$ and $p\in C[x^{1/v}]$ with $v\in \bN\setminus\{0\}$. The \emph{generalized indicial polynomial} of $L$ at $0$ and $\exp(p(x^{-1}))$, denoted by $\ind_{0,\exp(p(x^{-1}))}(L)$, is defined as the indicial polynomial of $\tilde L = \exp(-p(x^{-1}))\,L\,\exp(p(x^{-1}))\in C[x^{1/v},x^{-1/v}][D]$. When $p(x^{-1})=0$, the generalized indicial polynomial coincides with the classical indicial polynomial of $L$ at $0$.
\end{definition}

Let $L$ be an operator of order $r$. For each $\xi \in C\cup\{\infty\}$, $S_\xi(L)$ is defined as before:
\begin{equation}\label{EQ: S_xi(L)}
S_\xi (L):= \sum_{j=1}^r e_j(\xi) - \frac{r(r-1)}{2},
\end{equation} 
but now $e_j(\xi)$ are the generalized local exponents of $L$ at $\xi$, see~\cite[p. 297]{Corel01} or~\cite[\S 3.3]{Aldossari20} for their definition (they are the roots of the generalized indicial polynomials of $L$ at $\xi$). Let 
\begin{equation}\label{EQ: I_xi(L)}
	I_\xi(L) := 2 \sum_{1\leq i<j\leq r} \deg(p_i - p_j)
\end{equation}
where $p_i$ are the exponential parts of $L$ at $\xi$. If $\xi$ is a regular singularity, then $I_\xi(L) =0$. The generalized Fuchs relation (see~\cite{Bertrand99, Corel01} or~\cite[\S 3.5]{Aldossari20}) states that
\begin{equation}\label{EQ: fuchs rel generalized}
	\sum_{\xi\in \Sing(L) \cup \{\infty\}} (S_\xi (L) - \frac{1}{2} I_{\xi}(L))= -r(r-1).
\end{equation}

\begin{example}\label{ex:2*2=4,Q}
	The operator $Q=(x-1)D^2 +xD-1\in \bC(x)[D]$ is non-Fuchsian. It has two singularities at $1$ and $\infty$. The point $1$ is regular and apparent. Its indicial polynomial \[\ind_1(Q) = s(s-2)\] has degree 2, corresponding to the series solutions $1+O(x - 1), (x - 1)^2+O((x - 1)^3)$ with local exponents $0$, $2$ respectively. Its Newton polygon at $0$ has one edge of slop $1$ and width $2$. The point $\xi=\infty$ is irregular. Its generalized indicial polynomials are \[\ind_{\infty, \exp(0)}(Q) = s+1,\quad\ind_{\infty, \exp(x)}(Q) = -s,\] 
	corresponding to series solutions $x, \exp(x)$ with generalized local exponents are $-1$, $0$ respectively. Its Newton polygon at $\infty$ (i.e. the Newton polygon at $0$ of the operator obtained by substituting $x\mapsto x^{-1}$ into $Q$) has two edges: one of slop $1$ and width $1$, and one of slope $2$ and width $1$. For the operator $Q$,
	\begin{align*}
		S_1(Q) &= 0+2 - 1 = 1,\\
		S_\infty(Q) &= -1 + 0 -1 = -2,\\
		I_\infty(Q) &= 2\cdot 1= 2,
	\end{align*}
	and the generalized Fuchs relation~\eqref{EQ: fuchs rel generalized} reduces to
	\[1 - 2 -  1 = - 2.\]
\end{example}

\section{The colon space}\label{sec:colonspace}
For two ideals $I,J$ of a commutative ring $R$, the set
\[(I:J) := \{r\in R \mid r J \subseteq I\}\]
is called the \emph{ideal quotient} (or \emph{colon ideal}) of $I$ by $J$. If $R$ is a polynomial ring in several variables over~$C$, one can compute a Gr\"obner basis of a colon ideal, see details in~\cite[\S 4.4]{CoxJohnDonal15-book}. 

In this section, let $K$ be a ring extension of $C$. Then $K$ is naturally a $C$-algebra, i.e., a $C$-vector space equipped with a compatible ring structure. A typical choice for $K$ is the ring of formal power series $C[[x]]$ or the field of formal Laurent series $C((x))$. As an analog of colon ideals, we introduce the following notion.
\begin{definition}\label{def: colon space}
	Let $V$ be a $C$-vector subspace of $K$ and $U$ be a subset of $K$. The \emph{colon space} of \,$V$ by $U$ in $K$ is defined as the set
	\[(V: U) :=\{h\in K \mid hU \subseteq V\},\]
	which is a $C$-vector subspace of $K$.
\end{definition}
The solution space of a symmetric quotient is contained in the corresponding colon space.

\begin{lemma}\label{lem: V(Q) subseteq V(M):V(L)}
	Let $L,Q, M\in C(x)[D]$ be such that $L\otimes Q$ is a right factor of $M$. Let $F$ be a differential ring extension of $C(x)$ and let $(V_F(M): V_F(L))$ be the colon space in $F$. Then $V_{F}(Q)\subseteq (V_F(M): V_F(L))$.
\end{lemma}
\begin{proof}
	For any $h\in V_F(Q)$, we have $h\in F$ and $h V_F(L)\subseteq V_F(M)$ by the definition of symmetric products. This implies  $h\in (V_F(M): V_F(L))$ by the definition of the colon space.
\end{proof}

The colon space satisfies the following basic properties, analogous to those of colon ideals in the polynomial ring, (see~\cite[Proposition 13 and Theorem 14 in \S 4.4]{CoxJohnDonal15-book}).
\begin{prop}\label{prop: colon properties}  Let $g$ be an invertible element in $K$ and let $V, U, U_1,\ldots, U_r$ be $C$-vector subspaces of~$K$. Then
	\begin{enumerate}[label=(\roman*)]
		\item\label{it:colon1}  $(V: \{g\})=\{f/g \mid f\in V\}$.
		\item\label{it:colon2}  If $\{f_1,\ldots,f_n\}$ generates $V$ as a $C$-vector space, then $(V: \{g\})$ is generated by $\left\{\frac{f_1}{g},\ldots,\frac{f_n}{g}\right\}$.
		\item\label{it:colon3} If $\{g_1,\ldots,g_r\}$ generates $U$ as a $C$-vector space, then $(V: U)=(V: \{g_1,\ldots,g_r\})$.
		\item\label{it:colon4} $(V: (\sum_{j=1}^rU_j))=\bigcap_{j=1}^r (V: U_j)$. 
	\end{enumerate}
\end{prop}

\begin{proof}
	\begin{enumerate}[label=(\roman*)]
		\item follows from the definition of the colon space.
		\item By item~\ref{it:colon1}, $\frac{f_i}{g}$ belongs to $(V: \{g\})$ for all $1\leq i \leq n$. For any $h\in (V: \{g\})$, again by~\ref{it:colon1}, there exists $f\in V$ such that $h=\frac{f}{g}$. Since $f$ is a $C$-linear combination of $f_1,\ldots,f_n$, it follows that $h=\frac{f}{g}$ is a $C$-linear combination of $\frac{f_1}{g},\ldots, \frac{f_n}{g}$.
		\item Since $\{g_1,\ldots, g_r\}$ is a subset of $U$, by the definition of colon spaces we obtain that $(V: U)$ is a subset of $(V: \{g_1,\ldots, g_r\})$. Conversely, suppose $h\in (V: \{g_1,\ldots,g_r\})$. Then by definition, $hg_i\in V$ for all $1\leq i \leq n$. Every element $g$ in~$U$ can be written as a $C$-linear combination $g=\sum_{i=1}^rb_ig_i$ with $b_1,\ldots, b_r \in C$. Thus $hg = \sum_{i=1}^r b_i (hg_i)\in V$ because $V$ is a $C$-vector space. By the arbitrariness of~$g$, we have $hU\subseteq V$. Therefore $h\in (V: U)$. 
		\item For every $h\in V$, we have $h\in (V: (\sum_{j=1}^rU_j))\Leftrightarrow h(\sum_{j=1}^rU_j)\subseteq V\Leftrightarrow
		 \forall \,1\leq j\leq r, \,hU_j\subseteq V\Leftrightarrow h\in~\bigcap_{j=1}^r(V: U_j)$.\qedhere
	\end{enumerate}
\end{proof}

\begin{corollary}\label{cor: colon}
	Let $V=\Span_C\{f_1,\ldots,f_n\}$ and $U=\Span_C\{g_1,\ldots g_r\}$ be two $C$-vector subspaces of $K$. If $g_1,\ldots,g_r$ are invertible elements of $K$, then
	\[(V: U)= \bigcap_{i=1}^r (V: \{g_i\}) = \bigcap_{i=1}^r \Span_C\left\{\frac{f_1}{g_i}, \ldots, \frac{f_n}{g_i}\right\}.\]
\end{corollary}
\begin{proof}
	For each $i=1,\ldots, r$, let $U_i = \Span_C\{g_i\}$. Then $U= U_1 + \cdots + U_r$. By Proposition~\ref{prop: colon properties}, we have
	\[(V: U) = (V: (U_1+\cdots+U_r)) \overset{\ref{it:colon4}}{=}  \bigcap_{i=1}^r (V: U_i) \overset{\ref{it:colon3}}{=}  \bigcap_{i=1}^r (V: \{g_i\}) \overset{\ref{it:colon2}}{=}  \bigcap_{i=1}^r \Span_C\left\{\frac{f_1}{g_i}, \ldots, \frac{f_n}{g_i}\right\}. \qedhere\]
\end{proof}



\begin{example}\label{Ex:colon}
	Let $L=x^2 D^2 -2xD+2, Q=x^3D^3-3x^2D^2+6x-6\in\bC(x)[D]$, and \[M=L\otimes Q = x^4D^4 - 8xD^3 +36x^2D^2 -96xD+120.\] 
	We consider the solution spaces of these operators in $K=\bC((x))$:
	\begin{equation}\label{eq:example-colon-sol}
		V(L)=\Span_\bC\{x,x^2\},\quad V(Q)=\Span_\bC\{x,x^2,x^3\},\quad V(M) = \Span_\bC\{x^2,x^3,x^4,x^5\}.
	\end{equation}
	By Corollary~\ref{cor: colon}, we have
	\begin{align*}
		(V(M): V(L)) &=(V(M): \{x\})\cap(V(M): \{x^2\})\\
		&=\Span_\bC\{x,x^2,x^3,x^4\}\cap\Span_\bC\{1,x,x^2,x^3\}\\
		&= \Span_\bC\{x,x^2,x^3\}.
	\end{align*}
	By Lemma~\ref{lem: V(Q) subseteq V(M):V(L)}, the solution space $V(Q)$ of the quotient $Q$ is contained in the colon space $(V(M): V(L))$. In this example, we have equality $V(Q)=(V(M): V(L))$.
\end{example}
	 %

\section{Symmetric division algorithm}\label{sec:divisionalg}
\subsection{The maximal symmetric quotient}
Unlike polynomial division,  symmetric division may admit infinitely many quotients. Even the order of the quotient may not be unique. If $L,M\in C(x)[D]$ are of positive order such that $M=L\otimes Q$ for some $Q\in C(x)[D]$, it is known that 
\begin{equation}\label{eq:upper-lower-orders-sym}
	\ord(L)+\ord(Q)-1\leq\ord(M)\leq \ord(L)\ord(Q).
\end{equation}
Since $\ord(L)\neq0$, this implies
\begin{equation}\label{eq:upper-lower-orders}
	(\ord(M)/\ord(L))\leq \ord(Q)\leq \ord(M)-\ord(L)+1.
\end{equation}
Therefore, only finitely many orders can appear for the symmetric quotients. 

\begin{example}
	Let $L,Q,M\in \bC(x)[D]$ be the same as in Example~\ref{Ex:colon}. Let 
	\[Q_\alpha=(-\alpha x^2-2x^3)D^2 +(2\alpha x +6x^2)D+(-2\alpha-6x),\] where $\alpha\in\bC$. These operators have enough solutions in $\bC((x))$. The solution spaces $V(L)$, $V(Q)$, $V(M)$ are listed in~\eqref{eq:example-colon-sol}. The solution space of $Q_\alpha$ in $\bC((x))$ is
	\[V(Q_\alpha)=\Span_\bC\{x,x^3+\alpha x^2\}.\]
	Since
	\begin{align*}
		V(M)&=\Span_\bC\{\,gh\mid g\in V(L)\text{ and } h\in V(Q)\}=\Span_\bC\{\,gh\mid g\in V(L)\text{ and } h\in V(Q_\alpha)\},
	\end{align*}
	it follows that
	\[M=L\otimes Q=L\otimes Q_\alpha \quad\text{for all $\alpha\in \bC.$}\] 
	Thus, for all $\alpha\in \bC$, the operator $Q_\alpha$ is a second-order quotient of $M$ by $L$ with respect to symmetric product. The operator $Q$ is also a quotient but of order three. By~\eqref{eq:upper-lower-orders}, $Q$ attains the maximal order among all symmetric quotients of 
	$M$ by $L$. In this example, 
	\[V(Q_\alpha)\subseteq V(Q)\quad\text{and}
	\quad Q= R_\alpha Q_\alpha,\]
	where $R_\alpha=(-\frac{x}{\alpha+2x}D+\frac{3}{\alpha + 2x})\in\bC(x)[D]$. Therefore $Q$ is a left multiple of $Q_\alpha$ for all $\alpha\in \bC$. 
\end{example}




	If $L=0$, then $L\otimes Q = 0$ for any $Q\in C(x)[D]$. If $L\in C(x)\setminus\{0\}$, then  $L\otimes Q= 1$ for any $Q\in C(x)[D]$. To avoid such degenerate cases, we consider only operators of positive order in symmetric division. In this section, we prove that maximal-order symmetric quotients are unique up to left multiplication by nonzero rational functions.
	

\begin{prop}\label{prop: maximal quasi quotient}
	Let $L,M\in C(x)[D]$ be of positive order. Then there exists a unique primitive operator $Q\in C(x)[D]$ of maximal order such that $L\otimes Q$ is a right factor of $M$. Moreover, this operator $Q$ is a least common left multiple (lclm) of all operators $P$ such that $L\otimes P$ is a right factor of $M$.
\end{prop}
\begin{proof}
	Let \[\delta:=\max\{\ord(P)\mid P\in C(x)[D],\, L\otimes P \text{ is a right factor of } M\}.\] This set of orders is non-empty because for any $L,M\in C(x)[D]$ of positive order, $L\otimes 1 = 1$ is a trivial right factor of $M$. By~\eqref{eq:upper-lower-orders}, if $M_0:=L\otimes P$ is a right factor of $M$ for some $P\in C(x)[D]$, then 
	\[\ord(P)\leq \ord(M_0)-\ord(L)+1\leq \ord(M)-\ord(L)+1.\] Hence the set of $\ord(P)$ is finite and therefore it admits a maximum value $\delta$.
	
	Let $P_1,P_2\in C(x)[D]$ be operators of order $\delta$ such that  $L\otimes P_1$ and $L\otimes P_2$ are right factors of~$M$. Suppose that both $P_1$ and $P_2$ are primitive operators but $P_1\neq P_2$. By Corollary~\ref{cor:lclm-quotient-2}, we obtain that  $L\otimes P$ is a right factor of $M$, where $P=\lclm(P_1,P_2)$. However, $\ord(P)>\ord(P_1)=\delta$, which contradicts the maximality of~$\delta$. So there exists a unique primitive operator $Q\in C(x)[D]$ of order $\delta$ such that $L\otimes Q$ is a right factor of $M$.
	
	Suppose that $L\otimes P$ is a right factor of $M$ for some $P\in C(x)[D]$. We need to show that $Q$ is a left multiple of $P$. Using Corollary~\ref{cor:lclm-quotient-2} again, we obtain that  $P_0:=\lclm(P,Q)$ is a right factor of $P$. If $Q$ is not a left multiple of $P$, then $\ord(P_0)>\ord(P) = \delta$, which contradicts the maximality of~$\delta$.
\end{proof}

\begin{definition}
	Let $L,M\in C(x)[D]$ be of positive order. The \emph{(global) quasi-symmetric quotient} of $M$ by $L$, denoted by $\qsquo(M,L)$, is defined as the unique primitive operator $Q\in C(x)[D]$ of maximal order such that $L\otimes Q$ is a right factor of $M$. 
\end{definition}

By Proposition~\ref{prop: maximal quasi quotient}, a quasi-symmetric quotient exists and is unique. In particular, let $L, M\in C(x)[D]$ be of positive order, and suppose that $M=L\otimes P$ for some $P\in C(x)[D]$. Then there exists a quasi-symmetric quotient $Q$ of $M$ by $L$. By Corollary~\ref{cor:lclm-quotient-2}, $L\otimes Q=L\otimes \lclm(Q,P)=\lclm(L\otimes Q, L\otimes P) =M$.
So~$Q$ is a symmetric quotient, i.e., the unique primitive operator of maximal order such that $L\otimes Q = M$.


\subsection{Overview of the algorithm}

Given two operators $L,M\in C(x)[D]$ of positive order, we want to find the quasi-symmetric quotient $Q$ of $M$ by $L$. To do this, we first search for the power series solutions of $Q$. Then we try to recover the coefficients of $Q$ from its solution space by solving a linear system over $C$. In this section, we work with solution spaces in the field of formal Laurent series $C((x))$. After change of variables, we may assume without loss of generality that $0$ is an ordinary point of both $L$ and $M$. Then a new upper bound for the orders of symmetric quotients is given as follows.

\begin{prop}\label{prop:order}
	Let $L, M\in C(x)[D]$ be of positive order such that $L\otimes Q$ is a right factor of $M$ for some $Q\in C(x)[D]$. Let $V(L), V(Q),V(M)$ be the solution spaces of $L,Q,M$ in $C((x))$, respectively. 	Let $(V(M): V(L))$ be the colon space in $C((x))$. If $0$ is an ordinary point of $L$ and $M$, then
	\begin{enumerate}[label=(\roman*)]
		\item\label{it:quo1} $0$ is either an ordinary point or an apparent singularity of $Q$.
		\item\label{it:quo2} $V(Q)\subseteq (V(M): V(L))$.
		\item\label{it:quo3} $\ord(Q)=\dim_CV(Q)\leq \dim_C(V(M): V(L))$,
	\end{enumerate}

\end{prop}
\begin{proof}
	\begin{enumerate}[label=(\roman*)]
		\item 	Suppose on the contrary that $0$ is a singularity of $Q$ but not an apparent singularity. Then the solution space of $Q$ in $C[[x]]$ has dimension strictly less than $\ord(Q)$. This implies that $Q$ has a solution $h \in F\setminus C[[x]]$, where $F$ is a differential ring extension of $C[[x]]$ with constant field $C$. 
		
		Since $0$ is an ordinary point of $L$, it follows from Lemma~\ref{lem:ordinary_fund_sys} that $L$ has a formal power series solution $g$ of the form $g=1 + O(x)$. Then $g$ is an invertible element in $C[[x]]$.
		
		By the definition of symmetric product, $gh\in F$ is a solution of $L\otimes Q$. Then $gh$ is also a solution of~$M$, because $L\otimes Q$ is a right factor of $M$. Since $0$ is an ordinary point of $M$, it follows from Lemma~\ref{lem:ordinary_fund_sys} that $M$ has $\ord(M)$ linearly independent solutions in $C[[x]]$. The dimension of the solution space of $M$ in $F$ can not exceed its order. Therefore $f := gh\in C[[x]]$, which implies that $h = fg^{-1} \in C[[x]]$. This leads to a contradiction.
		\item follows from Lemma~\ref{lem: V(Q) subseteq V(M):V(L)} by taking $F=C((x))$.
		\item By~\ref{it:quo2}, we obtain $\dim_CV(Q)\leq \dim_C(V(M): V(L))$. Since $0$ is an ordinary point of $M$ and $L$, it follows from~\ref{it:quo1} that $0$ is either an ordinary point of an apparent singularity of~$Q$. In either case, we have $\ord(Q) = \dim_CV(Q)$. \qedhere
	\end{enumerate}
\end{proof}

Before presenting the symmetric division algorithm, we give an example to illustrate its main idea.
\begin{example}\label{ex:2*3=4}
	Let $L,M$ be two operators in $\bC(x)[D]$:
	\begin{align*}
		&L =  (x^2 - 2x + 2)(x - 1)^2D^2 +2x(x - 1)(x - 2)D+2\\
		&M= L\otimes P = (x^2 - 2x + 2)(x - 1)^4D^4-8(x - 1)^3D^3+36(x - 1)^2D^2-96(x -1)D+120,
	\end{align*}
	for some unknown $P\in \bC(x)[D]$. We want to compute a symmetric quotient of $M$ by $L$.
 
We first compute the power series solutions of $L$ and $M$ at the ordinary point $x=0$. The operator $M$ has four linear independent solutions in $\bC[[x]]$:
	\begin{align*}
		f_1&=1 - x - \frac{1}{2}x^2 +0x^3+ \frac{1}{4}x^4 +\frac{1}{4}x^5 +  O(x^6),\\
		f_2&= x -  x^2 -  \frac{1}{2}x^3 + 0x^4+ \frac{1}{4}x^5 + O(x^6),\\
		f_3&= x^2 - x^3 -  \frac{1}{2}x^4 +0x^5 + O(x^6),\\
		f_4&= x^3 - x^4 -  \frac{1}{2}x^5 + O(x^6).
	\end{align*}
	The operator $L$ has two linearly independent solutions in $\bC[[x]]$:
	\begin{align*}
		g_1&=1 +0x- \frac{1}{2}x^2 - \frac{1}{2}x^3 - \frac{1}{4}x^4 +0x^5+ O(x^6),\\
		g_2&= x +0x^2- \frac{1}{2}x^3 - \frac{1}{2}x^4 - \frac{1}{4}x^5 + O(x^6).
	\end{align*}
	Then $V(M)=\Span_\bC\{f_1,f_2,f_3,f_4\}$ and $V(L) = \Span_\bC\{g_1,g_2\}$. 	
	
	Let $(V(M): V(L))$ be the colon space in~$\bC((x))$. By Corollary~\ref{cor: colon}, we have \[(V(M): V(L))= (V(M): \{g_1\})\cap(V(M): \{g_2\})=\Span_\bC\left\{\frac{f_1}{g_1},\frac{f_2}{g_1},\frac{f_3}{g_1},\frac{f_4}{g_1}\right\}\cap\Span_\bC\left\{\frac{f_1}{g_2},\frac{f_2}{g_2},\frac{f_3}{g_2},\frac{f_4}{g_2}\right\}.\]
	Our algorithm in Section~\ref{sec:order} finds $\dim_\bC(V(M): V(L))=3$, with a basis given by
	\begin{align*}
		h_1&=1 - x + 0x^2+0x^3+0x^4+O(x^5),\\
		h_2&= x-x^2 +0x^2+0x^3+0x^4+O(x^5),\\
		h_3&= x^2- x^3 +0x^2+0x^3+0x^4+O(x^5).
	\end{align*}
	By Proposition~\ref{prop:order}.\ref{it:quo2}, if $M=L\otimes Q$ for some $Q\in \bC(x)[D]$, then the solution space $V(Q)$ in $\bC((x))$ is a subspace of $(V(M): V(L))$. We search for an operator $Q$ of order three such that $V(Q)=(V(M): V(L))$. Our algorithm in Section~\ref{sec:degree} shows that the degrees of coefficients of $Q$ are at most $d_0=27$. In this example, we find 
	\[Q=(x-1)^3D^3 - 3(x-1)^2D^2 + 6(x-1)D - 6\]
	It can be verified that $M=L\otimes Q$. Since $\ord(Q)=\dim_\bC(V(M): V(L))$, Proposition~\ref{prop: maximal quasi quotient} implies that $Q$ has maximal order among all symmetric quotients of $M$ by $L$. In fact, $P=	(x - 1)^2D^2 - 3(x - 1)D + 3$ is another symmetric quotient of $M$ by $L$, but of order two. Moreover, $Q= RP$ is a left multiple of $P$, where $R=\frac{1}{(x-1)^2}((x-1)^3D - 2x^2 + 4x - 2)$.
\end{example}

\begin{definition}
	Let $L,M\in C(x)[D]$ be of positive order and let $\xi\in C$ be an ordinary point of both $L$ and $M$. Let $V_\xi(L)$ and $V_\xi(M)$ denote the respective solution spaces of $L$ and $M$ in $C((x-\xi))$ and let $(V_\xi(M): V_\xi(L))$ be the colon space in $C((x-\xi))$. A \emph{local quasi-symmetric quotient} of $M$ by $L$ at $x=\xi$, denoted by $\qsquo(M,L,x=\xi)$, is defined as a primitive operator $Q\in C(x)[D]$ such that the solution space of $Q$ in $C((x-\xi))$ equals $(V_\xi(M): V_\xi(L))$, and $L\otimes Q$ is a right factor of $M$. 
\end{definition}

Throughout the remainder of this paper, all colon spaces $(V(M): V(L))$ are taken in $C((x))$.

\begin{lemma}\label{lem: local-global}
	Let $L,M\in C(x)[D]$ be of positive order, with $0$ an ordinary point of both $L$ and $M$. Then a local quasi-symmetric quotient of $M$ by $L$ at $0$, if it exists, is the global quasi-symmetric quotient.
\end{lemma}
\begin{proof}
	Suppose there exists a local quasi-symmetric quotient of $M$ by $L$ at $0$, denoted by $Q$. Then $L\otimes Q$ is a right factor of $M$. It suffices to show that $Q$ has maximal order. By Proposition~\ref{prop:order}.\ref{it:quo3}, for any operator $P\in C(x)[D]$ such that $L\otimes P$ is a right factor of $M$, we have $\ord(P)\leq \dim_C(V(M): V(L))$. By the definition of the local quasi-symmetric quotient, $\ord(Q)=\dim_C(V(M):V(L))$. Therefore $\ord(P)\leq \ord(Q)$. Thus $Q$ is the global quasi-symmetric quotient of $M$ by $L$.
\end{proof}

The next lemma leads to an equivalent description of local quasi-symmetric quotients, given in the subsequent corollary.
\begin{lemma}\label{lem: quotient - equivalent condition}
	Let $L,M\in C(x)[D]$ be of positive order, with $0$ an ordinary point of both $L$ and $M$. For any operator $Q\in C(x)[D]$, if $V(Q)\subseteq (V(M): V(L))$, then
	$L\otimes Q$ is a right factor of $M$ if and only if $\dim_CV(Q)=\ord(Q)$.
\end{lemma}
\begin{proof}
		For $Q\in C(x)[D]$, suppose that $V(Q)\subseteq (V(M):V(L))$ and $L\otimes Q$ is a right factor of $M$, Proposition~\ref{prop:order}.\ref{it:quo3} implies $\dim_CV(Q)=\ord(Q)$.
	
	For the converse, suppose $V(Q)\subseteq (V(M):V(L))$ and $\dim_CV(Q)=\ord(Q)$. By Definition~\ref{def: colon space} of the colon space, if $g\in V(L)$ and $h\in V(Q)$, then 
	\begin{equation}\label{eq: qquo proof 1}
		gh=hg\in h V(L)\subseteq V(M),
	\end{equation}
	i.e., $gh$ is a solution of $M$ in $C((x))$. Since $0$ is an ordinary point of $L$, it follows from Lemma~\ref{lem:ordinary_fund_sys} that $\dim_CV(L) = \ord(L)$. By the assumption, we have $\dim_CV(Q)=\ord(Q)$. Then, by Lemma~\ref{lem:sol-lcml-sym} on properties of the symmetric product, we obtain 
	\begin{equation}\label{eq: qquo proof 2}
		V(L\otimes Q) =\Span_C\{gh\mid g\in V(Q), \, h \in V(L)\}\quad \text{and}\quad \dim_C(V(L\otimes Q)) = \ord(L\otimes Q).
	\end{equation}
	Combining \eqref{eq: qquo proof 1} and~\eqref{eq: qquo proof 2} yields that the solution space of $L\otimes Q$ in $C((x))$ has full dimension and is a subspace of~$V(M)$. Thus Lemma~\ref{lem:sol-rightfactor}.\ref{it:rf2} implies that $L\otimes Q$ is a right factor of $M$.
\end{proof}

\begin{corollary}\label{cor: quotient - equivalent condition}
	Let $L,M\in C(x)[D]$ be of positive order, with $0$ an ordinary point of both $L$ and $M$. 
	Then a primitive operator $Q\in C(x)[D]$ is a local quasi-symmetric quotient of $M$ by $L$ at $0$ if and only if 
	\[V(Q) = (V(M): V(L))\quad\text{and}\quad \dim_CV(Q)=\ord(Q).\]
\end{corollary}
\begin{proof}
	It follows from the definition of local quasi-symmetric quotients and Lemma~\ref{lem: quotient - equivalent condition}.
\end{proof}

If we can compute a basis $\{h_1,\ldots,h_\delta\}$ of $(V(M): V(L))$ to sufficient precision $k$ and have an upper bound on the degrees of the coefficients of an order-$\delta$ quasi-symmetric quotient $Q$ of $M$ by $L$, we can make an ansatz for $Q$ and set up a linear system $Q\cdot h_j=O(x^{k-\delta})$ for all $j=1,\ldots,\delta$. When a nonzero solution is found, we check whether $\ord(Q)=\delta$ and $L\otimes Q$ is a right factor of $M$. If so, $Q$ is a local quasi-symmetric quotient. The procedure is summarized in the following algorithm. Our symmetric division algorithm is inspired by Algorithm 1 in~\cite{BostanRivoalSalvy2024}, which finds the minimal annihilator of a D-finite power series. Its correctness and termination arguments are very similar.
\begin{algorithm}\label{alg:local at zero}
	INPUT:	$L,M\in C(x)[D]$ of positive order, with $0$ an ordinary point of both $L$ and $M$.\\
	OUTPUT: a local quasi-symmetric quotient $Q\in C(x)[D]$ of $M$ by $L$ at $x=0$, or \emph{None} if no such $Q$ exists.

	\step 11 {{\bfseries function}} $\textsc{QuasiSymmetricQuotientAtZero}(M,L)$
	\step 22 set $r:=\ord(L)$.
	\step 32 $\delta := \textsc{ColonSpaceDimension}(M,L)$.
	\step 42 if $\delta=0$, {{\bfseries return}} $Q:=1$.
	\step 52 if $\delta>0$, 
	\step 63 set $k_0 := \max \Z_N - \frac{r(r-1)}{2}$ and $k=k_0+1$, where $N = M\otimes L^{\otimes (r-1)}$ and $\Z_N$ is defined in~\eqref{eq:Z_N}.
	\step 73 $d_0 := \textsc{BoundDegreeOfQuasiSymmetricQuotientCoeffs}(M,L,\delta)$.
	\step 83 while true do:
	\step 94 $\{h_1,\ldots,h_\delta\}:=\textsc{ColonSpaceBasis}(M, L, k)$.
	\step {10}4 $d := \min\{d_0,\lfloor (k-\delta)/(\delta+1)\rfloor \}$.
	\step {11}4 $Q:=\textsc{ApproximantAnnihilator}(h_1,\ldots,h_\delta; d, \ldots, d;k)$
	\step {12}4 if $Q = \emptyset$ and $k\geq (\delta+1)(d_0+1)+\delta$, then {{\bfseries return}} \emph{None}.
	\step {13}4 if $Q\neq \emptyset$,
	\step {14}5 if $\ord(Q)=\delta$ and $L\otimes Q$ is a right factor of $M$, then {{\bfseries return}} $Q$.

	\step {15}4 $k:=2k$.
\end{algorithm}

The above algorithm relies on several sub-algorithms that we now summarize.

\begin{description}[font=\normalfont\textsc,leftmargin=0cm]
	\item[\textsc{ColonSpaceDimension}] computes the dimension of the colon space $(V(M): V(L))$ in $C((x))$ over $C$, see Section~\ref{sec:order}.
	\item[\textsc{ColonSpaceBasis}] computes a $C$-vector space basis of the colon space $(V(M): V(L))$ in $C((x))$ at precision $k$, see Section~\ref{sec:order}.
	\item[\textsc{BoundDegreeOfQuasiSymmetricQuotientCoeffs}] returns an upper bound on the degree of each of the coefficients of an order-$\delta$ global quasi-symmetric quotient of $M$ by $L$, see Section~\ref{sec:degree}.
	\item[\textsc{ApproximantAnnihilator}] takes as input $\delta$ power series $(h_1,\ldots,h_\delta)$ that are the truncations at precision $k$; $\delta+1$ nonnegative integers $(s_0,\ldots,s_{\delta})$ and the precision $k$. It returns a primitive operator $Q=q_0 + q_1 D + \cdots + q_{\delta}D^\delta$ with $q_i\in C[x]$ and $\deg(q_i)\leq s_i$ such that
	\begin{equation}\label{eq: approx ann}
		(q_0 + q_1 D + \cdots + q_{\delta}D^\delta)\cdot h_j= O(x^{k-\delta})
	\end{equation}
	for all $j=1,\ldots,\delta$; or returns $\emptyset$ if there is no such an operator $Q$. Since only one annihilator is required, this can be computed by solving a linear system. If one wants to compute all such annihilators, one can first compute a basis $B(x)=(B_{i,j})_{0\leq i,j\leq \delta}\in C[x]^{(\delta+1)\times (\delta+1)}$ of the $C[x]$-module
	\begin{equation}\label{eq: approx basis}
		\{(q_0,q_1,\ldots,q_\delta)\in C[x]^{1\times (\delta+1)}\mid q_0 h_j + q_1h_j' + \cdots q_\delta h_j^{(\delta)} = O(x^{k-\delta})\}
	\end{equation}
	for all $j=1,\ldots,\delta$, in \emph{shifted Popov form}~\cite{Popov72,vanBarelBultheel1992,BeckermannLabahnVillard99,JeannerodNeigeerVillard20}  with shift vector $(-s_0, \ldots, -s_\delta)$. This implies that any solution of~\eqref{eq: approx basis} with degrees bounded $(s_0,\ldots,s_\delta)$ is a linear combination of the rows of $B$ whose index $i$ satisfies $\deg(B_{i,i})\leq s_i$. Efficient algorithms to compute such bases are known~\cite{JeannerodNeigeerVillard20}.

\end{description}

\begin{theorem}
	Algorithm~\ref{alg:local at zero} terminates and is correct.
\end{theorem}
\begin{proof}
	\begin{enumerate}
		\item (Correctness assuming termination) In line 3, set $\delta := \dim_C(V(M): V(L))$. If $\delta=0$, then  in line~4, $Q:=1$ satisfies that $V(Q)=\emptyset = (V(M): V(L))$ and $L\otimes Q =1$ is a right factor of $M$. If $\delta>0$, Theorem~\ref{thm: order computation} in Section~\ref{sec:order} guarantees that, when $k> k_0$ in line~6, the truncation of $(V(M): V(L))$ at precision $k$ has the same dimension as $(V(M): V(L))$. In line~9, we compute a basis $\{h_1,\ldots,h_\delta\}$ of $(V(M):V(L))$ at precision~$k$. All series $h_j, h_j', \ldots, h_j^{(\delta)}$ are known at precision $k-\delta$ for each $j=1,\ldots,\delta$. By Lemma~\ref{lem: local-global}, the degree bound $d_0$ in line~7 for the coefficients of global quasi-symmetric quotients also applies in the local case. In line~12, the condition $k\geq (\delta+1)(d_0+1)+\delta$ ensures that in line 11, we have $\lfloor (k-\delta)/(\delta+1)\rfloor \geq d_0+1$ and hence~$d=d_0$. If \textsc{ApproximantAnnihilator} returns empty with the given degree bounds on the degrees of the coefficients in line 12, then there is no operator $Q$ of order $\delta$ such that $L\otimes Q$ is a right factor of $M$. Otherwise, in line 14, if there exists an operator $Q\in C(x)[D]$ of order~$\delta$ such that $L\otimes Q$ is a right factor of $M$. Then by Proposition~\ref{prop:order}.\ref{it:quo3} and line~3, we get $V(Q)= (V(M): V(L))$. Thus $Q$ is a local quasi-symmetric quotient of $M$ by $L$ at $0$.
		\item (Termination) The only possible source on non-termination in Algorithm~\ref{alg:local at zero} is the loop where $k$ is doubled every time no $Q$ in line~14 is found. Let $U_k$ be the set of all solutions $(q_0,q_1,\ldots,q_\delta)\in C[x]^{(\delta+1)}$ of~\eqref{eq: approx ann} with degrees bounded by $(d_0,\ldots, d_0)$. Then for all $k>k_0$, $U_k$ is a $C$-vector space of finite dimension and $U_{k+1}\subseteq U_k$. Thus there exists $k'$ such that $U_{k'}$ is the intersection of all $U_k$ for $k>k'$. Any operator $Q=q_0+q_1D+\cdots+q_\delta D^\delta$ returned by \textsc{ApproximantAnnihilator} in line 11 for $k>k'$ has the property that $Q\cdot h_j = O(x^p)$ for all $p\geq k-\delta$ and thus annihilates~$(V(M): V(L))$. Since $(V(M): V(L))$ has dimension $\delta$ and $Q=q_0+q_1D+\cdots +q_\delta D^\delta$ has order at most $\delta$, it follows that \[V(Q)=(V(M): V(L)) \quad \text{and}\quad\ord(Q)=\delta=\dim_C(V(Q)).\]
		
		By Lemma~\ref{lem: quotient - equivalent condition}, $L\otimes Q$ is a right factor of $M$. This guarantees termination of Algorithm~\ref{alg:local at zero}. \qedhere 
	\end{enumerate}
\end{proof}

In practice, line~6 of Algorithm~\ref{alg:local at zero} is not optimal especially when there exists a local quasi-symmetric quotient. Let $\{g_1,\ldots,g_r\}$ be a basis of $V(L)$ and let $n=\ord(M)$. One may first try smaller values of~$k$, without computing $\Z_N$. For any $k$, one can compute a basis $\{h_1,\ldots,h_{\delta_k}\}$ of the truncated space \[\bigcap_{i=1}^r T_k(V(M):\{g_i\}),\] which contains $T_k\left(\bigcap_{i=1}^r (V(M):\{g_i\})\right)=T_k( V(M):V(L))$, and therefore provides a good approximation of $(V(M):V(L))$ at precision~$k$. If $k>\Z_N-\frac{r(r-1)}{2}$, we will prove in Proposition~\ref{thm: order computation} that these two truncated spaces are equal and their dimension is the dimension of the colon space $(V(M):V(L))$. If~$k\geq n-1$, we will show in Lemma~\ref{lem: w_i} that $T_k$ is an injective map from $\bigcap_{i=1}^rT_{k+1}(V(M):\{g_i\})$ to $\bigcap_{i=1}^rT_{k}(V(M):\{g_i\})$ and hence $\delta_{k} \geq \delta_{k+1}$. Since the chain of dimensions $\delta_{n-1}\geq \delta_n\geq \dots$ is decreasing and eventually stabilizes at $\delta$. Thus $\delta_k\geq \delta$ for all $k\geq n-1$.

Using an upper bound on the degrees of the coefficients of an order-$\delta_k$ quasi-symmetric quotient, one can try to search for a possible local quasi-symmetric quotient. If, for a sufficiently large $k$ $(k\geq n-1)$, we find an operator $Q$ of order~$\delta_k$ such that $L\otimes Q$ is a right factor of $M$, then Proposition~\ref{prop:order}.\ref{it:quo3} implies that $\delta_k\leq\delta$. Thus $\delta_k=\delta$ and $Q$ is a local quasi-symmetric quotient of $M$ by $L$. 

If $\delta_k>0$ and one wants to prove that no local quasi-symmetric quotient exists, then in our approach it is necessary to compute $\Z_N$, or at least an upper bound  $Z_N$, to determine the exact dimension $\delta$ of the colon space $(V(M):V(L))$.
\section{Three special cases}\label{sec:special}
In this section, we show that in certain special cases, the order bound for symmetric quotients in Proposition~\ref{prop:order} is sharp, and a local quasi-symmetric quotients always exists. Moreover, in these cases, the following algorithm returns the global quasi-symmetric quotient. The correctness of this algorithm follows from Lemma~\ref{lem: local-global}. In our experiments, for random operators $M,L\in C(x)[D]$ such that $M=L\otimes P$ for some unknown $P\in C(x)[D]$, the algorithm always finds the global quasi-symmetric quotient of $M$ by~$L$. However, in the general case, a theoretical proof or counterexample remains open.


\begin{algorithm}\label{alg:global}
	INPUT: $L,M\in C(x)[D]$ of positive order.\\
	OUTPUT: the global quasi-symmetric quotient $Q\in C(x)[D]$ of $M$ by $L$, or \emph{Fail} (which does not imply nonexistence; existence is guaranteed by Proposition~\ref{prop: maximal quasi quotient}).
	
	\step 11 {{\bfseries function}} $\textsc{QuasiSymmetricQuotient}(M,L)$
	\step 22 choose an arbitrary ordinary point $\xi$ of $L$ and $M$.
	\step 32 transform $L$ and $M$ by substituting $x\to x+\xi$ to obtain $L_\xi$ and $M_\xi$.
	\step 42 compute $Q_\xi = \textsc{QuasiSymmetricQuotientAtZero}(M_\xi,L_\xi)$.
	\step 52 if $Q_\xi = \rm{None}$, then {{\bfseries return}} \emph{Fail}.
	\step 62 otherwise, transform $Q_\xi$ back via $x\to x-\xi$ to obtain $Q$.
	\step 72 {{\bfseries return}} $Q$.
\end{algorithm}


Based on Corollary~\ref{cor: quotient - equivalent condition}, we give another equivalent description of local quasi-symmetric quotients.
\begin{lemma}\label{lem: quotient - equivalent condition 2}
	Let $L,M\in C(x)[D]$ be of positive order, with $0$ an ordinary point of both $L$ and $M$. Then there exists a local quasi-symmetric quotient $Q\in C(x)[D]$ of $M$ by $L$ at $0$ if and only if the colon space $(V(M):V(L))$ satisfies the following two conditions:
	\begin{enumerate}[label=(\alph*)]
		\item\label{it:assump1} every series $h\in (V(M): V(L))$ is D-finite;
		\item\label{it:assump2} for each $h\in (V(M): V(L))$, if $Q_h$ is a minimal annihilator of $h$ in $C(x)[D]$, then \begin{equation*}
			V(Q_h)\subseteq (V(M):V(L))\quad\text{and}\quad \dim_CV(Q_h) =\ord(Q_h).
		\end{equation*}
 	\end{enumerate}
\end{lemma}
\begin{proof}
	Suppose that $Q\in C(x)[D]$ is a local quasi-symmetric quotient of $M$ by $L$ at $0$. Then by Corollary~\ref{cor: quotient - equivalent condition}, there exists $Q\in C(x)[D]$ such that $V(Q)=(V(M): V(L))$ and $\ord(Q)=\dim_CV(Q)$. So~every $h\in (V(M): V(L))$ is D-finite because $Q$ is an annihilator of $h$. For each $h\in (V(M): V(L))$, let $Q_h$ be a minimal annihilator of $h$. Then $Q_h$ is a right factor of $Q$. Thus, by Lemma~\ref{lem:sol-rightfactor}.\ref{it:rf3}, $P$ satisfies the required condition in~\ref{it:assump2}.

	For the converse, since $V(M)$ and $V(L)$ are finite-dimensional $C$-vector subspaces of $C((x))$, it follows from Corollary~\ref{cor: colon} that $(V(M): V(L))$ is also finite-dimensional. Let $\delta:=\dim_C(V(M): V(L))$. If $\delta=0$, then take $Q=1$. If $\delta>0$, suppose that $\{h_1,\ldots,h_\delta\}$ is a basis of $(V(M): V(L))$. For each $1\leq j\leq \delta$, by the condition~\ref{it:assump1}, the series $h_j$ is D-finite. Let $Q_j\in C(x)[D]$ be a minimal annihilator of~$h_j$. We will show that $Q := \lclm_{j=1}^\delta Q_j$ is a desired operator.
		
	By the condition~\ref{it:assump2}, we know
	\[V(Q_j)\subseteq (V(M):V(L))\quad\text{and}\quad \dim_CV(Q_j) = \ord(Q_j).\]
	By Lemma~\ref{lem:sol-lcml-sym}.\ref{it:sol-lclm} on properties of the lclm, we have
	\[V(Q) = V(Q_1)+\cdots + V(Q_\delta)\subseteq (V(M): V(L)) \quad\text{and}\quad \dim_CV(Q)=\ord(Q).\]
	Since $\{h_1,\ldots,h_\delta\}$ is a basis of $(V(M): V(L))$, we obtain that $V(Q)=(V(M): V(L))$. By Corollary~\ref{cor: quotient - equivalent condition}, $Q$ is a local quasi-symmetric quotient of $M$ by $L$ at $0$.
\end{proof}

\subsection{The hyperexponential case}


Let $F$ be a differential field extension of $C(x)$. Recall that a nonzero element $f\in F$ is said to be \emph{hyperexponential} over $C(x)$ if its logarithmic derivative $(D\cdot f)/f$ is a rational function in $C(x)$. Equivalently, $f$ is \emph{hyperexponential} if it is a nonzero solution of some first-order operator $u D - v$ with $u,v\in C(x)$, $u\neq 0$. If $f\in F$ is hyperexponential, then its inverse $f^{-1}$ is also hyperexponential. 

%

\begin{lemma}\label{lem: min ann of gh}
	Let $L\in C(x)[D]$ be a first-order operator and let $g$ be a nonzero solution of $L$ in $C((x))$. Let $h\in C((x))$, $h\neq 0$ be D-finite and let $Q\in C(x)[D]$ be a minimal annihilator of $h$. Then $L\otimes Q$ is a minimal annihilator of $gh$.
\end{lemma}
\begin{proof}
	By definition of symmetric products, $L\otimes Q$ is an annihilator of $gh$. To show it is of minimal order, suppose for contradiction that there exists an annihilator $M\in C(x)[D]$ of $gh$ with $\ord(M)>\ord(L\otimes Q)$. Since $L$ is of first order, we write $L= u D - v$, where $u,v\in C(x)$, $u\neq 0$. Then $R:= (uD+v)$ is an annihilator of $g^{-1}$. So the symmetric product $R \otimes M$ is also an annihilator of $h = g^{-1}gh$. Since~$\ord(R)=1$, by~\eqref{eq:upper-lower-orders-sym} we have $\ord(R\otimes M) = \ord(M)>\ord(L\otimes Q)=\ord(Q)$.
	This contradicts the assumption that $Q$~is a minimal annihilator of $h$.
\end{proof}

\begin{remark}
	If $L$ is not of first order, Lemma~\ref{lem: min ann of gh} may not hold. For example, the second-order operator $L=-2x^2D^2 + xD -1 \in \bC(x)[D]$ is a minimal annihilator of $g=x+\sqrt{x}$, and $Q=-2x^2D^2+xD-1$ is a minimal annihilator of $h=x-\sqrt{x}$. The product $gh=x^2-x$ is hyperexponential and hence it has a minimal annihilator $(-x+x^2)D+(1-2x)$ of order $1$. The symmetric product $L\otimes Q=2xD^3-3x^2D^2+6xD-6$ is an annihilator of $gh$, but not a minimal annihilator.
\end{remark}

\begin{theorem}\label{thm: hyper case}
		Let $L,M\in C(x)[D]$ be of positive order. If $L=\lclm_{i=1}^I L_i$ where the $L_i$ are first-order operators in $C(x)[D]$, then Algorithm~\ref{alg:global} returns the global quasi-symmetric quotient of $M$ by~$L$.
\end{theorem}
\begin{proof}
	In line~3 of Algorithm~\ref{alg:global}, after the shift $x\mapsto x+\xi$ , $L_\xi$ is still the lclm of some first-order operators. So we may further assume that $0$ is an ordinary point of $L$ and $M$. It suffices to show that in line~4 there exists a local quasi-symmetric quotient of $M$ by $L$ at $0$. Thus we only need to verify the conditions \ref{it:assump1} and~\ref{it:assump2} in Lemma~\ref{lem: quotient - equivalent condition 2}.
	
	Since $L$ is the lclm of several first-order operators $L_i$, it follows from~Lemma~\ref{lem:sol-lcml-sym}.\ref{it:sol-lclm} that \[V(L) = \Span_{C}\{g_1,\ldots,g_I\},\] 
	where $g_i$ is a nonzero solution of $L_i$ for $1\leq i\leq I$. By Corollary~\ref{cor: colon} we get
	\begin{equation}\label{eq: hyper-quotient}
		(V(M): V(L)) =\bigcap_{i=1}^I (V(M): \{g_i\})\subseteq (V(M):\{g_1\}) = \{ f g_1^{-1} \mid f\in V(M)\}.
	\end{equation}
	
	\ref{it:assump1}: Since $g_1$ is hyperexponential, its inverse $g_1^{-1}$ is D-finite. Since the product of any two D-finite functions is also D-finite, it follows from~\eqref{eq: hyper-quotient} that every element of $(V(M):V(L))$ is D-finite. 
	
	\ref{it:assump2}: For each $h\in (V(M): V(L))$, let $Q_h$ be a minimal annihilator of $h$ in $C(x)[D]$. For each $1\leq i\leq I$, Lemma~\ref{lem: min ann of gh} implies that $L_i\otimes Q_h$ is a minimal annihilator of $g_ih$, . By the definition of the colon space, $g_ih\in V(M)$, i.e., $M$ is an annihilator of $g_ih$. Since $L_i\otimes Q_h$ has minimal order among all annihilators of $g_ih$, the operator $L_i\otimes Q_h$ is a right factor of $M$ for all $1\leq i\leq I$. Since~$L=\lclm_{i=1}^I L_i$, it follows from Corollary~\ref{cor:lclm-quotient-2} that $L\otimes Q_h$ is a right factor of $M$. Thus, by Proposition~\ref{prop:order}, we get $V(Q_h)\subseteq (V(M): V(L))$ and $\ord(Q_h)=\dim_CV(Q_h)$. 
\end{proof}

\subsection{The C-finite case}
In the shift case, the Hadamard quotient of two linear recurrence sequences with constant coefficients was studied in~\cite{CorvajaZannier1998,CorvajaZannier2002,Zannier2005}. Kauers and Zeilberger~\cite{kauersZeilberger17} presented a complete factorization algorithm for linear recurrence equations with constant coefficients with respect to symmetric product, based on polynomial factorization. In the differential case, factorization theory for exponential polynomials was initiated by Ritt~\cite{Ritt1927factorization} and extended in a general setting by MacCol~\cite{MacColl1935factorization} and later by Everest and van der Poorten~\cite{EverestVanderPoorten1997factorisation}.  

In this section, we consider the symmetric quotient of two linear differential operators with constant coefficients. A nonzero power series $f\in C[[x]]$ is called \emph{C-finite} if there exists $L\in C[D]\setminus\{0\}$ such that $L\cdot f=0$. We will show that for any input with constant coefficients, Algorithm~\ref{alg:global} always returns the global quasi-symmetric quotient with constant coefficients, instead of rational coefficients. So $d_0=0$ is a natural degree bound for the quotient. In Section~\ref{sec:order}, the computation of degree bounds for symmetric quotients in the general case is based on the C-finite case. In practice, our symmetric division algorithm with constant coefficients relies only on linear algebra, although polynomial factorization is used in theory.
\begin{lemma}[{\!\!\cite[Section 6.1]{Ince1926}}]\label{lem: C-finite - structure thm}
	Let $L= \ell_rD^r + \ell_{r-1}D^{r-1} + \cdots + \ell_1D+\ell_0\in C[D]$ be a differential operator with constant coefficients of order $r$. Suppose that $P$ factorizes as 
	\begin{equation}\label{eq: factorization}
		L= \ell_r(D-\alpha_1)^{\mu_1}(D-\alpha_2)^{\mu_2}\cdots (D-\alpha_\rho)^{\mu_I}
	\end{equation}
	where $\mu_1,\ldots, \mu_I\in\bN\setminus\{0\}$ and the roots $\alpha_1, \ldots ,\alpha_I\in C$ are distinct. Then the elements
	\[x^j \exp(\alpha_ix) \quad (i= 1, \ldots,I, j = 0, \ldots, \mu_i - 1)\]
	are $r$ linearly independent solutions of $L$ in $C[[x]]$.
\end{lemma}

The following lemma is an immediate consequence of Lemmas \ref{lem: C-finite - structure thm} and~\ref{lem:sol-lcml-sym}.
\begin{lemma}\label{lem: C-finite sym prod}
	Let $L=\prod_{i=1}^s(D-\alpha_i)^{\mu_i}$ and $Q = \prod_{j=1}^t (D-\beta_j)^{\lambda_j}$ be factorizations in $C[D]$, where $\mu_i,\lambda_j\in \bN\setminus\{0\}$, $\alpha_i,\beta_j\in C$,  the $\alpha_i$ (resp. the $\beta_j$) are pairwise distinct. Then 
	\[L\otimes Q = \lcm_{i=1}^s\lcm_{j=1}^t (D-\alpha_i-\beta_j)^{\mu_i+\lambda_j-1}.\]
\end{lemma}

A variant of Lemma~\ref{lem: min ann of gh} for annihilators of products is as follows.

\begin{lemma}\label{lem: min ann of gh - C-finite}
	Let $g=x^j\exp(\alpha x)$ with $j\in \bN$, $\alpha\in C$. Let 
	\begin{equation}\label{eq: C-finite h}
		h = \sum_{i=1}^\rho u_i(x) \exp(\theta_ix),
	\end{equation} 
	where $\theta_1,\ldots, \theta_\rho\in C$ are distinct, and $u_1,\ldots, u_\rho \in C[x] \setminus \{0\}$ with $\deg_x(u_i) = d_i$. If $L$ and $Q$ are minimal annihilators of $g$ and $h$ in $C[D]$, respectively, then $L\otimes Q$ is a minimal annihilator of $gh$ in $C[D]$.
\end{lemma}
\begin{proof}
	Since the $\theta_i$ are distinct, it is known that $L:=(D-\alpha)^{j+1}$ and $Q:=\prod_{i=1}^\rho(D-\theta_i)^{d_i+1}$ are minimal annihilators of $g$ and $h$ in $C[D]$, respectively. By Lemma~\ref{lem: C-finite sym prod}, we get $L\otimes Q =  \prod_{i=1}^\rho (D-\theta_i-\alpha)^{d_i+j+1}$ and hence it is a minimal annihilator of $gh=\sum_{i=1}^\rho x^ju_i(x)\exp((\theta_i+\alpha)x)$.
\end{proof}

\begin{theorem}\label{thm: C-finite case}
	Let $L,M\in C[D]$ be of positive order. Then Algorithm~\ref{alg:global} returns the global quasi-symmetric quotient $Q$ in $C[D]$ of $M$ by~$L$, rather than in $C(x)[D]$.
\end{theorem}

\begin{proof}
	The proof is similar to that of Theorem~\ref{thm: hyper case}.
	We write $L = \prod_{i=1}^s (D-\alpha_i)^{\mu_i}$ with $\mu_i\in\bN\setminus\{0\}$ and distinct $\alpha_i \in C$. By Lemma~\ref{lem: C-finite - structure thm}, we have \[V(L) = \Span_{C}\{x^j\exp(\alpha_i x) \mid i=1,\ldots,s, j=0,\ldots,\mu_i-1\}.\] By Corollary~\ref{cor: colon}, we get
	\begin{align}
		(V(M): V(L)) &= \bigcap_{i=1}^s\bigcap_{j=0}^{\mu_i-1}(V(M): \{x^j\exp(\alpha_i x)\})\label{eq: C-finite-quotient-1}\\
		&\subseteq (V(M): {\exp(\alpha_1x)})=\{f \exp(-\alpha_1x)\mid f\in V(M)\}.\label{eq: C-finite-quotient-2}
	\end{align}
	
	
	For a fixed $h\in (V(M): V(L))$, since $M\in C[D]$, combining \eqref{eq: C-finite-quotient-2} and Lemma~\ref{lem: C-finite - structure thm} yields that  $h$ can be written in the form $\sum_{i=1}^\rho u_i(x)\exp(\theta_ix)$ as in~\eqref{eq: C-finite h}. Thus, $h$ is C-finite. 
	
	Let $Q_h$ be a minimal annihilator of $h$ in $C[D]$. For each $g_{i,j}:=x^j\exp(\alpha_ix)$, $L_{i,j} := (D-\alpha_i)^{j+1}$ is its minimal annihilator in $C[D]$. By Lemma~\ref{lem: min ann of gh - C-finite}, $L_{i,j}\otimes Q_h$ is a minimal annihilator of $g_{i,j}h$ in $C[D]$. By~the definition of the colon space, $M$ is also an annihilator of $g_{i,j}h$. Thus $L_{i,j}\otimes Q_h$ is a right factor of $M$. Since $L=\lcm_{i=1}^s\lcm_{j=0}^{\mu_i-1}L_{i,j}$, it follows from Corollary~\ref{cor:lclm-quotient-2} that $L\otimes Q_h$ is a right factor of $M$. Thus, by Proposition~\ref{prop:order}, we get $V(Q_h)\subseteq (V(M): V(L))$ and $\ord(Q)=\dim_CV(Q)$. 
	
	By literally adapting Lemma~\ref{lem: quotient - equivalent condition 2} to the C-finite case, we obtain the existence of a local quasi-symmetric quotient $Q\in C[D]$ of $M$ by $L$ at $0$. Hence, by Lemma~\ref{lem: local-global}, Algorithm~\ref{alg:global} returns a global quasi-symmetric quotient $Q\in C[D]$ of $M$ by~$L$.
\end{proof}


\subsection{The algebraic case}

A series $f\in C((x))$ is called \emph{algebraic} if there exists a nonzero polynomial $m(x,y)\in C(x)[y]$ such that~$m(x, f)=0$. Let $\overline{C(x)}$ be the algebraic closure of $C(x)$. For an algebraic series $f\in C((x))$ with minimal polynomial $m\in C(x)[y]$, the roots of $m$ in $\overline{C(x)}$ are \emph{conjugates} of $f$. An \emph{algebraic function field} $E = C(x)[y]/\<m>$ is a field extension of the rational function field $C(x)$ of finite degree, where $m$ is an irreducible polynomial in $C(x)[y]$. The usual derivation $'$ on $C(x)$ extends uniquely to the field $E=C(x)[y]/\<m>$. For any $f\in E$, all its derivatives $f,f',f'', \ldots$ belong to $E$. Hence every algebraic function is D-finite. An operator $L\in C(x)[D]$ has \emph{only algebraic solutions}~\cite{Singer1979algebraic,Singer02,kauersKoutschanThibaut23,BostanSalvySinger2025} if there is a differential field $E = C(x)[y]/\<m>$ such that the solution space $V_E(L)$ of $L$ in $E$ has dimension~$\ord(L)$. For an algebraic function, its minimal annihilator has only algebraic solutions~\cite{Singer1979algebraic,Singer02}. Moreover, the solution space of its minimal annihilator is spanned by all the conjugates of $f$, see the following lemma.
\begin{lemma}\label{lem: min ann - algebraic}
	Let $L\in C(x)[D]$ with $r=\ord(L)$. Assume that $L$ has a nonzero solution $f$ which is algebraic over $C(x)$. Let $E$ be the algebraic extension of $C(x)$ generated by all conjugates of $f$.
	\begin{enumerate}[label=(\roman*)]
		\item\label{it:alg1} All conjugates of $f$ are solutions of $L$.
		\item\label{it:alg2} If $L$ is a minimal annihilator of $f$, then the solution space of $L$ in $E$ has dimension $r$ and is spanned by all conjugates of $f$.
	\end{enumerate}
\end{lemma}
\begin{proof}
	The field $E$ is a Galois extension of $C(x)$. Let $\Gal(E/C(x))$ be the Galois group of $E$ over $C(x)$. 
	\begin{enumerate}[label=(\roman*)]
		\item The set $\{\tau(f)\mid \tau \in \Gal(E/C(x))\}$ consists of all conjugates of $f$. Since $L\in C(x)[D]$ has coefficients in $C(x)$, for any $\tau\in \Gal(E/C(x))$, we have $L\cdot \tau(f) = \tau(L\cdot f) = 0$.
		\item Let $f_1,\ldots, f_s$ be all conjugates of $f$. For each $f_i$, the operator $L_i = f_i D - f_i'$ is a minimal annihilator of $f_i$ in $E[D]$. We take $\bar L :=\lclm(L_1,\ldots, L_s)$. For any automorphism $\tau \in \Gal(E/C(x))$, $\tau(\bar L)$ is obtained by applying $\tau$ to the coefficients of $\bar L$. Since taking the least common left multiple (lclm) is commutative, it follows that
		\[\tau(\bar L) = \lclm(\tau(L_1), \ldots, \tau(L_s)) = \lclm (L_1,\ldots, L_s)=\bar L.\]
		This implies that $\bar L$ has coefficients in $C(x)$. Since $\dim_CV_E(L_i) = \ord(L_i)=1$ for all $i=1,\ldots, s$, by Lemma~\ref{lem:sol-lcml-sym}.\ref{it:sol-lclm} we have
		\[V_E(\bar L) = V_E(L_1) + \cdots + V_E(L_s)=\Span_C\{f_1,\ldots,f_s\}\quad\text{and}\quad \dim_CV_E(\bar L) = \ord(\bar L).\] 
		By the item $(i)$, all conjugates of $f$ must be solutions of its annihilator. Thus $\bar L$ is a minimal annihilator of $f$ in $C(x)[D]$. Since $L$ is also a minimal annihilator of $f$ in $C(x)[D]$, we get $L = u \bar L$ for some $ u \in C(x)\setminus\{0\}$. So $L$ and $\bar L$ have the same solution space.\qedhere
	\end{enumerate}
\end{proof}

\begin{theorem}
	Let $L,M\in C(x)[D]$ be of positive order. If $L$ and $M$ have only algebraic solutions, then Algorithm~\ref{alg:global} returns the global quasi-symmetric quotient of $M$ by~$L$.
\end{theorem}
\begin{proof}
	In line~3 of Algorithm~\ref{alg:global}, after the shift $x\mapsto x+\xi$ , $M_\xi$ and $L_\xi$ still have only algebraic solutions. So we may further assume that $0$ is an ordinary point of $L$ and $M$. Similar to the proof of Theorem~\ref{thm: hyper case}, we only need to verify the conditions \ref{it:assump1} and~\ref{it:assump2} in Lemma~\ref{lem: quotient - equivalent condition 2}.
	
	Since $0$ is an ordinary point of $L$ and $M$, it follows from Lemma~\ref{lem:ordinary_fund_sys} that $L$ has $r=\ord(L)$ linearly independent solutions $\{g_1,\ldots,g_r\}$ in $C[[x]]$ and $M$ has $n=\ord(M)$ linearly independent solutions $\{f_1,\ldots,f_n\}$ in $C[[x]]$. By the assumption, $g_1,\ldots, g_r, f_1,\ldots,f_n$ are algebraic over $C(x)$. Let $E$ be the algebraic extension of $C(x)$ generated by $g_1,\ldots, g_r,f_1,\ldots, f_n$ and all their conjugates. Then $E$ is a Galois extension of $C(x)$. Let $\Gal(E/C(x)) = \{\tau_1,\ldots,\tau_t\}$ be the Galois group of $E$ over $C(x)$. 
	
	By Lemma~\ref{lem: min ann - algebraic}, for each $1\leq i\leq r$ and each $1\leq j \leq t$, the element $\tau_j(g_i)$ is still a solution of $L$. So the set $\{\tau_j(g_i)\mid 1\leq i\leq r, 1\leq j\leq t\}$ is also a generating set of the solution space~$V(L)$.
	Note that the computation of the colon space $(V(M): V(L))$ does not depend on the choice of generating sets of $V(L)$. Using Corollary~\ref{cor: colon}, we obtain that
	\begin{equation}\label{eq:colon-alg}
		(V(M): V(L)) = \bigcap_{i=1}^r \bigcap_{j=1}^t (V(M): \{\tau_j(g_i)\}) = \bigcap_{i=1}^r \bigcap_{j=1}^t \Span_C\left\{\frac{f_1}{\tau_j(g_i)}, \ldots,\frac{f_n}{\tau_j(g_i)}\right\}\subseteq E\cap C((x)).
	\end{equation}
	
	\ref{it:assump1}: By~\eqref{eq:colon-alg}, every power series in $(V(M): V(L))$ is algebraic and hence it is D-finite.
	
	\ref{it:assump2}: For a fixed element $h\in (V(M): V(L))$, let $Q_h$ be a minimal annihilator of $h$ in $C(x)[D]$. By~\eqref{eq:colon-alg}, we have $h\tau_j(g_i) \in V(M)$ for all $1\leq i \leq r$ and all $1\leq j\leq t$. Then Lemma~\ref{lem: min ann - algebraic}.\ref{it:alg1} implies that 
	\begin{equation}\label{eq:tau}
		\tau(h\tau_j(g_i)) = \tau(h) (\tau\circ\tau_j)(g_i)\in V(M) \quad\text{for all }\tau\in \Gal(E/C(x)).
	\end{equation}
	For each fixed $\tau\in \Gal(E/C(x))$, $\tau\circ \tau_j$ for $j=1,\ldots,t$ run through all elements of the group $\Gal(E/C(x))$. It follows from \eqref{eq:colon-alg} and~\eqref{eq:tau} that $\tau(h)\in (V(M): V(L))$ for all $\tau \in \Gal(E/C(x))$. Then all conjugates of~$h$ belong to $(V(M): V(L))$. Therefore, by Lemma \ref{lem: min ann - algebraic}.\ref{it:alg2}, we have $V(Q_h)\subseteq (V(M): V(L))$ and $\dim_CV(Q_h) = \ord(Q_h)$.
\end{proof}
\section{Truncation and intersection of power series subspaces}\label{sec:truncation}

For an operator $N\in C[x][D]$, let $V$ be the solution space of $N$ in the ring of formal power series $C[[x]]$. Then $V$ is a $C$-vector space of finite dimension, at most $\ord(N)$. By linear algebra, the intersection of several $C$-vector subspaces of $V$ is still a $C$-vector space of finite dimension. However, the elements of $V$ are formal power series with infinitely many coefficients. To compute a basis of the intersection space, or to determine its dimension, we shall work with truncated power series to approximate the intersection. Since power series solutions of $N$ satisfy a recurrence relation, the required truncation precision can be determined by the following proposition. This result will be used in the next section to determine the dimension of the colon space $(V(M):  V(L))$, and to compute a basis for it. 

For a $C$-vector subspace $W$ of $C[[x]]$,  and $g\in C[[x]]$, we write $gW$ to denote the set $\{gf \mid f\in W\}$. Then $gW$ is also a $C$-vector subspace of $C[[x]]$.

\begin{prop}\label{prop: intersection dim}
	Let $V$ be the solution space of an operator $N\in C[x][D]$ in $C[[x]]$. Let $W_1,\ldots, W_r$ be $C$-vector subspaces of $C[[x]]$ such that $gW_1,\ldots, gW_r$ are $C$-vector subspaces of $V$, where $g=\sum_{i=\mu}^\infty b_ix^i$ with $b_i\in C$ and $b_\mu\neq 0$. Let $k_0=\max \Z_N - \mu$, where $\Z_N$ is defined in~\eqref{eq:Z_N}. Then for all $k> k_0$,
	\begin{enumerate}[label=(\roman*)]
		\item $\dim_{C}\left(\, \bigcap_{i=1}^rW_i \right) = \dim_C\left(\, \bigcap_{i=1}^r T_k(W_i) \right)$
		\item $T_k\left(\, \bigcap_{i=1}^r W_i \right) = \bigcap_{i=1}^r T_k(W_i)$
	\end{enumerate}
\end{prop}

To prove this proposition, we need several lemmas.  Let $C^\bN$ denote the set of all infinite sequences $(a_0,a_1,a_2,
\ldots)$ with $a_i\in C$.  A formal power series can be viewed as a sequence in $C^\bN$ via its sequence of coefficients. Under this identification, $C[[x]]$ is isomorphic to $C^\bN$ as a ring. In particular, $C^\bN$ is also a $C$-vector space of infinite dimension. The \emph{coefficient vector} of a formal power series $\varphi=\sum_{i=0}^\infty a_ix^i\in C[[x]]$ is defined as the column vector: \[\coeff(\varphi):=(a_0,a_1,a_2,\ldots)^T\in C^\bN,\]
where the $i$-th entry of $\coeff(\varphi)$ is the coefficient of $f$ in $x^{i-1}$. Let $\varphi_1,\ldots,\varphi_\rho$ be several formal power series in $C[[x]]$. We write each $\varphi_j = \sum_{i=0}^\infty a_{i,j}x^i$ with $a_{i,j}\in C$. The \emph{coefficient matrix} of $\varphi_1,\ldots,\varphi_{\rho}$ is defined as 
\[\coeff(\varphi_1,\ldots,\varphi_\rho):= (\coeff(\varphi_1),\ldots,\coeff(\varphi_\rho))=	\begin{pmatrix}
	a_{0,1}&\cdots&a_{0,\rho}\\
	a_{1,1}&\cdots&a_{1,\rho}\\
	\vdots& & \vdots
\end{pmatrix}\in C^{\bN\times \rho}.\]

For a matrix $A=(a_{i,j})\in C^{\bN\times \rho}$, its \emph{row space} is defined by 
\[\row(A):=\Span_C\{(a_{i,1},\ldots,a_{i,\rho})\in C^\rho \mid i\in\bN\},\]
which is a $C$-vector subspace of $C^\rho$. The \emph{kernel} of $A$ is defined by
\[\ker(A) :=\left\{(s_1,\ldots,s_\rho)\in C^{\rho} \,\middle |\, \sum_{j=1}^\rho a_{i,j}s_j= 0 \text{ for all $i\in \bN$} \right\}\]
If $A=\coeff(\varphi_1,\ldots,\varphi_\rho)\in C^{\bN\times \rho}$ with $\varphi_i\in C[[x]]$, then $\ker(A) = \{(s_1,\ldots,s_\rho)\in C^\rho \mid \sum_{j=1}^\rho s_i\varphi_i = 0\}$.

\begin{lemma}\label{lem:T_k(W)}
	Let $N\in C[x][D]$ and let $\{\varphi_1,\ldots,\varphi_\rho\}$ be a finite set of solutions of $N$ in $C[[x]]$. Let $k_0=\max \Z_N$. Then for all $k> k_0$, the row spaces of \[\coeff(\varphi_1,\ldots,\varphi_\rho)\quad \text{and} \quad \coeff(T_{k}(\varphi_1),\ldots, T_{k}(\varphi_\rho))\] 
	are equal.
\end{lemma}
\begin{proof}
	For each $1\leq j \leq \rho$, we write $\varphi_j= \sum_{i=0}^\infty a_{i,j}x^i$ with $a_{i,j}\in C$. For each $i\in \bN$, let $\va_i =(a_{i,1},\ldots, a_{i,\rho})\in C^\rho$. Then the row space of $\coeff(\varphi_1,\ldots,\varphi_\rho)$ is generated by $\{\va_i \mid i\in \bN\}$. The row space of $\coeff(T_{k}(\varphi_1),\ldots, T_{k}(\varphi_\rho))$ is generated by $\{\va_i \mid 0\leq i\leq k\}$. So it suffices to show by induction on~$k$ that $\va_k \in \Span_C\{\va_0,\ldots,\va_{k_0}\}$ for all $k\geq 0$. For $0\leq k\leq k_0$, it is clearly true. For $k>k_0$, we assume that $a_i\in \Span_C\{\va_0,\ldots,\va_{k_0}\}$ for $0\leq i\leq k-1$.  We write $N\cdot x^s=x^{s+\sigma_N}(q_0(s)+q_1(s)x+\cdots + q_t(s)x^t)$ in the form~\eqref{eq:Nx^s}. Since the $\varphi_j$ are power series solutions of $N$ and $k> k_0$, it follows from Lemma~\ref{lem:series-sol-rec} that 
	\[\va_k = -\frac{1}{q_0(k)}\left(\,\sum_{i=1}^{\min\{t,k\}}\va_{k-i}q_i(k-i)\right)\in \Span_C \{\va_{k-1},\ldots, \va_{k-\min\{t,k\}}\}\subseteq \Span_C \{\va_{k-1},\ldots, \va_0\}.\]
	By the induction hypothesis, we get $\va_k\in\Span_C\{\va_0,\ldots,\va_{k_0}\}$. This completes the proof.
\end{proof}

\begin{lemma}\label{lem: gW}
	Let $\{\varphi_1,\ldots,\varphi_\rho\}\subseteq C[[x]]$ and let $g=\sum_{i=\mu}^\infty b_ix^i \in C[[x]]$ with $b_i\in C$ and $b_\mu \neq 0$. Then for all $k\geq 0$, the row spaces of
	\[\coeff(T_k(\varphi_1),\ldots,T_k(\varphi_\rho))\quad \text{and} \quad \coeff(T_{k+\mu}(g\varphi_1),\ldots, T_{k+\mu}(g\varphi_\rho))\]
	are equal. Therefore the row spaces of 
	\[\coeff(\varphi_1,\ldots,\varphi_\rho)\quad \text{and} \quad \coeff(g\varphi_1,\ldots, g\varphi_\rho)\]
	are equal.
\end{lemma}
\begin{proof}
	For each $1\leq j\leq \rho$, we denote $\psi_j :=g\varphi_j$, and write $\varphi_j= \sum_{i=0}^\infty a_{i,j}x^i$ and $\psi_j= \sum_{i=0}^\infty c_{i,j}x^i$, where $a_{i,j}, c_{i,j}\in C$. For each $i\in\bN$, let $\va_i=(a_{i,1},\ldots,a_{i,\rho})$ and $\vc_i=(c_{i,1},\ldots,c_{i,\rho})$. It suffices to prove the first statement in the lemma. In other words, we need to show the claim that \[\Span_C\{\va_0,\va_1,\ldots,\va_k\}=\Span_C\{\vc_0,\vc_1,\ldots,\vc_{k+\mu}\} \text{ for all } k\geq 0.\]
	
	Since $\psi_j =g\varphi_j$, it follows that for all $1\leq j<\rho$, 
	\[c_{j,\ell} =0 \text{ for all } 0\leq \ell<\mu, \quad \text{and}\quad c_{j,\ell} = b_\ell a_{j,0} + b_{\ell-1}a_{j,1} + \cdots + b_\mu a_{j,\ell-\mu} \text{ for all } \ell\geq \mu.\]
	Therefore, for all $0\leq \ell < \mu$, $\vc_\ell =0$ and for all $\ell\geq \mu$, $\vc_\ell = b_\ell \va_{0} + b_{\ell-1}\va_{1} + \cdots + b_\mu \va_{\ell-\mu}$ is a $C$-linear combination of $\va_0, \va_1,\ldots, \va_{\ell-\mu}$. Thus $\Span_C\{\vc_0,\ldots,\vc_{k+\mu}\}\subseteq\Span_C\{\va_0,\ldots,\va_k\}$. Since $g$ is an invertible element in $C((x))$, we have $\varphi_j = g^{-1}\psi_j$, where $g^{-1} = \sum_{i=-\mu}^\infty \tilde b_i x^i\in C((x))$ with $\tilde b_i\in C$ and $\tilde b_{-\mu} = b_\mu^{-1}\neq 0$. Similarly, we get $\Span_C\{\va_0,\ldots,\va_k\}\subseteq \Span_C\{\vc_0,\ldots,\vc_{k+\mu}\}$. This proves the claim.
\end{proof}

\begin{lemma}\label{lem:dim r=2}
	Let $V$ be the solution space of an operator $N\in C(x)[D]$ in $C[[x]]$. Let $W_1,W_2$ be $C$-vector subspaces of $C[[x]]$ such that $gW_1,gW_2$ are $C$-vector subspaces of $V$, where $g=\sum_{i=\mu}^\infty b_ix^i\in C[[x]]$ with $b_\mu\neq 0$. Let $k_0=\max \Z_N - \mu$. Then for all $k> k_0$,
	\begin{enumerate}[label=(\roman*)]
		\item $\dim_{C}\left( W_1\cap W_2 \right) = \dim_C\left(\,T_k(W_1)\cap T_k(W_2) \right)$
		\item $T_k\left(W_1\cap W_2 \right) =  T_k(W_1)\cap T_k(W_2)$
	\end{enumerate}
\end{lemma}
\begin{proof}
	Since $V$ is a $C$-vector space of finite dimension, it follows that $W_1, W_2, gW_1, gW_2$ are also $C$-vector spaces of finite dimension. Let $\{\varphi_1,\ldots,\varphi_{\rho_1}\}$ and $\{\phi_1,\ldots,\phi_{\rho_2}\}$ be bases of $W_1$ and $W_2$, respectively. Then $\{g\varphi_1,\ldots,g\varphi_{\rho_1}\}$ and $\{g\phi_1,\ldots,g\phi_{\rho_2}\}$ are bases of $gW_1$ and $gW_2$, respectively. 
	
	Let $A=\coeff(\varphi_1,\ldots,\varphi_{\rho_1}, \phi_1,\ldots,\phi_{\rho_2})\in C^{\bN\times (\rho_1+\rho_2)}$  and let $A_k\in C^{(k+1)\times(\rho_1+\rho_2)}$ be the matrix consisting of the first $k+1$ rows of $A$. Then, for all $k\geq 0$, $\row(A_k)$ is equal to the row space of $\coeff(T_k(\varphi_1),\ldots,T_k(\varphi_{\rho_1}), T_k(\phi_1),\ldots,T_k(\phi_{\rho_2}))$.
	By linear algebra, we know
	\begin{equation}\label{eq: ker A}
		(s_1,\ldots,s_{\rho_1},t_1,\ldots,t_{\rho_2})\in \ker(A)\quad \Leftrightarrow\quad \sum_{j=1}^{\rho_1}s_i\varphi_i = -\sum_{j=1}^{\rho_2}t_j\phi_j\in W_1\cap W_2,
	\end{equation}
	and for all $k\geq 0$,
	\begin{equation}\label{eq:ker A_k}
		(s_1,\ldots,s_{\rho_1},t_1,\ldots,t_{\rho_2})\in \ker(A_k)\quad \Leftrightarrow\quad \sum_{j=1}^{\rho_1}s_jT_k(\varphi_j) = -\sum_{j=1}^{\rho_2}t_jT_k(\phi_j)\in T_k(W_1)\cap T_k(W_2).
	\end{equation}
	Moreover, we have $\dim_C(W_1\cap W_2)= \dim_C(\ker(A))$ and if $\dim_C(T_k(W_i))=\rho_i$ for $i=1,2$, then $\dim_C(T_k(W_1)\cap T_k(W_2))= \dim_C(\ker(A_k))$.
	\begin{enumerate}[label=(\roman*)]
		\item It suffices to show that for all $k> k_0$, $\ker(A) =\ker(A_k)$ and $\dim_C(T_k(W_i))=\rho_i$ for $i=1,2$. If this holds, then for all $k> k_0$, we have \[\dim_C(W_1\cap W_2) = \dim_C(\ker(A))=\dim_C(\ker(A_k))=\dim_C(T_k(W_1)\cap T_k(W_2)).\]
		To~prove this claim, we use the assumption that $gW_1, gW_2$ are subspaces of the solution space $V$ of the linear differential operator $N$. Let  $B=\coeff(g\varphi_1,\ldots,g\varphi_{\rho_1}, g\phi_1,\ldots,g\phi_{\rho_2})\in C^{\bN\times (\rho_1+\rho_2)}$ and let $B_k\in C^{(k+1)\times(\rho_1+\rho_2)}$ be the matrix consisting of the first $k+1$ rows of $B$. Then, for all $k\geq 0$, $\row(B_k)$ is equal to the row space of $\coeff(T_k(g\varphi_1),\ldots,T_k(g\varphi_{\rho_1}), T_k(g\phi_1),\ldots,T_k(g\phi_{\rho_2}))$. By the assumption, $\{g\varphi_j\}_{j=1}^{\rho_1}$, $\{g\phi_j\}_{j=1}^{\rho_2}$ are solutions of the linear differential operator $N$. Thus by Lemmas \ref{lem:T_k(W)} and~\ref{lem: gW} we obtain that for all $k>k_0$, 
		\begin{equation}\label{eq: row A = row A_k}
			\row(A) = \row(B) = \row(B_{k+\mu})=\row(A_{k}).
		\end{equation}
		Since the kernel of a matrix is determined by its row space, it follows from~\eqref{eq: row A = row A_k} that
		\begin{equation}\label{eq: ker A =ker A_k}
			\ker(A)  =\ker(A_k).
		\end{equation}
		
		Considering the first $\rho_1$ columns of $A$ and $A_k$, we obtain from~\eqref{eq: row A = row A_k} that \[\row (\coeff (\varphi_1,\ldots,\varphi_{\rho_1} )) = \row (\coeff (T_k(\varphi_1),\ldots,T_k(\varphi_{\rho_1}) )).\]
		Thus $\rho_1=\dim_C(W_1)=\dim_C(T_k(W_1))$ because the column rank of a matrix is equal to its row rank. Similarly, we get $\rho_2=\dim_C(W_2)=\dim_C(T_k(W_2))$. 	
%
%
%
%
		\item Since $W_1\cap W_2$ is a subspace of $W_1$, it follows that $T_k(W_1\cap W_2)\subseteq T_k(W_1)$. Similarly, we have $T_k(W_1\cap W_2)\subseteq T_k(W_2)$. Therefore $T_k(W_1\cap W_2)\subseteq T_k(W_1)\cap T_k(W_2)$.
		
		On the other hand, fix an arbitrary integer $k> k_0$, and suppose that $f\in T_k(W_1)\cap T_k(W_2)$. Then $f=\sum_{j=1}^{\rho_1}s_jT_k(\varphi_j) = - \sum_{j=1}^{\rho_2}t_jT_k(\phi_j)$ for some $s_j, t_j\in C$. By \eqref{eq:ker A_k} and~\eqref{eq: ker A =ker A_k}, we get $(s_1,\ldots,s_{\rho_1},t_1,\ldots,t_{\rho_2})\in \ker (A_k)=\ker(A)$. It then follows from~\eqref{eq: ker A} that
		\[g:= \sum_{j=1}^{\rho_1}s_j\varphi_j = - \sum_{j=1}^{\rho_2}t_j\phi_j\in W_1\cap W_2.\]
		Thus $f = T_k(g) \in T_k(W_1\cap W_2)$ because $s_j, t_j\in C$, and hence $T_k(W_1)\cap T_k(W_2)\subseteq T_k(W_1\cap W_2)$.\qedhere
	\end{enumerate}
\end{proof}
\begin{proof}[Proof of Propostion~\ref{prop: intersection dim}]
	\begin{enumerate}[label=(\roman*)]
		\item For $r=1$, since $gW_1$ is a subspace of the solution space $V$ of the operator~$N$ and $g=\sum_{i=\mu}^\infty b_ix^i\in C[[x]]$ with $b_\mu\neq 0$, it follows from Lemmas \ref{lem:T_k(W)} and~\ref{lem: gW} that for all $k> k_0$, 
		\[\dim_C(W_1) =\dim_C(gW_1)=\dim_C(T_{k+\mu}(gW_1))=\dim_C(T_k(W_1)).\]
		Here we use the fact that the column rank of a matrix is equal to its row rank. Suppose $r>1$. Note that $g(W_2\cap\cdots\cap W_r) = (gW_2)\cap \cdots\cap (gW_r)$  because $g$ is invertible in $C((x))$. Then by the assumption, $gW_1$ and $g(W_2\cap\cdots\cap W_r) $ are two $C$-vector subspaces of the solution space~$V$. By Lemma~\ref{lem:dim r=2}, we obtain 
		\[ \dim_{C}\left(W_1 \cap\left(W_2\cap \cdots \cap W_r\right)\right) =\dim_C\left(T_k\left(W_1 \cap \left( W_2 \cap \cdots \cap W_r\right)\right)\right).\]
		Thus $\dim_C\left(\,\bigcap_{i=1}^r W_i\right) =\dim_C\left(T_k\left(\,\bigcap_{i=1}^rW_i\right)\right)$ because  $\,\bigcap_{i=1}^r W_i = W_1\cap \left(W_2\cap \cdots\cap W_r\right)$. 
		\item We prove the statement by induction on $r$. For $r=1$, it is clearly true. For $r>1$, by Lemma~\ref{lem:dim r=2}, we have $T_k(W_1)\cap T_k(W_2) = T_k(W_1\cap W_2)$. By the induction hypothesis on $r-1$, we have
		\[T_k(W_1\cap W_2) \cap T_k(W_3) \cap \cdots \cap T_k(W_r) = T_k((W_1\cap W_2) \cap W_3\cap \cdots \cap W_r).\]
		Therefore $T_k(W_1)\cap T_k(W_2)\cap \cdots \cap T_k(W_r) = T_k(W_1\cap W_2\cap \cdots\cap W_r)$.\qedhere
	\end{enumerate}
\end{proof}

\section{Order bounds for symmetric quotients}\label{sec:order}
An upper bound for the orders of symmetric quotients is given in~\eqref{eq:upper-lower-orders}. A smaller upper bound is given in Proposition~\ref{prop:order} using the dimension of the colon space. To compute a basis of this colon space, we need the following notations.
  
\begin{convention}\label{con}
	Let $L,M\in C(x)[D]$ be of positive order, with $0$ an ordinary point of both $L$ and $M$. Let $\{g_1,\ldots,g_r\}$ be a basis of the solution space $V(L)$ in $C((x))$, where $r=\ord(L)$ and $g_i=x^{i-1}+O(x^i)$ for $i=1,\ldots, r$. Let $\{f_1,\ldots,f_n\}$ be a basis of the solution space $V(M)$ in $C((x))$, where $n=\ord(M)$ and $f_i=x^{i-1}+O(x^i)$ for $i=1,\ldots, n$. Let $(V(M): V(L))$ be the colon space in~$C((x))$. For~each $i=1,\ldots, r$, let $W_i = (V(M): \{g_i\})\cap C[[x]]$, where $(V(M): \{g_i\})$ is the colon space in $C((x))$.
\end{convention}
\begin{lemma}\label{lem: w_i}
	Let $V(L)$, $V(M)$ and $W_i$ be as in Convention~\ref{con}.  Then
	\begin{enumerate}[label=(\roman*)]
		\item\label{it: wi-1} $(V(M): V(L)) = \bigcap_{i=1}^r W_i$;
		\item\label{it: wi-2} for each $i=1,\ldots, r$, $W_i = \Span_C\left\{\frac{f_i}{g_i},\ldots, \frac{f_n}{g_i}\right\}$;
		\item\label{it: wi-3} for each $i=1,\ldots, r$ and for all $k\geq 0$, 
		
		$T_k(W_i) = \Span_C\left\{T_k\left(\frac{f_i}{g_i}\right), \ldots, T_k\left(\frac{f_n}{g_i}\right)\right\}=\Span_C\left\{T_k\left(\frac{T_{k+r-1}(f_i)}{T_{k+r-1}(g_i)}\right),\ldots, T_k\left(\frac{T_{k+r-1}(f_n)}{T_{k+r-1}(g_i)}\right)\right\}$.
		\item\label{it: wi-4} for all $k\geq n-1$, $T_k$ is an injective map from $\bigcap_{i=1}^rT_{k+1}(W_i)$ to $\bigcap_{i=1}^r T_{k}(W_i)$. In particular, $\dim_C\bigcap_{i=1}^rT_{k+1}(W_i)\leq \dim_C\bigcap_{i=1}^rT_{k}(W_i)$.
	\end{enumerate}

\end{lemma}
\begin{proof}
	\begin{enumerate}[label=(\roman*)]
		\item By Corollary~\ref{cor: colon}, we have $(V(P): V(L)) = \bigcap_{i=1}^r (V(P): \{g_i\})$. Since $g_1 = 1 + O(x)$ is invertible in $C[[x]]$, it follows from Proposition~\ref{prop: colon properties}.\ref{it:colon2} that $(V(P): \{g_1\}) =\Span_C\left\{\frac{f_1}{g_1},\ldots, \frac{f_n}{g_1}\right\}$ is a subspace of $C[[x]]$. Therefore $(V(P): V(L))\subseteq (V(P): \{g_1\}) \subseteq C[[x]]$ and hence
		\[(V(P): V(L)) = (V(P): V(L)) \cap C[[x]] = \bigcap_{i=1}^r ((V(P): \{g_i\})\cap C[[x]]) = \bigcap_{i=1}^r W_i.\]
		\item By Proposition~\ref{prop: colon properties}.\ref{it:colon2}, $W_i = \Span_C\left\{\frac{f_1}{g_i},\ldots,\frac{f_n}{g_i}\right\} \cap C[[x]]$. By Convention~\ref{con}, $\{f_1,\ldots,f_n\}$ are linearly independent over $C$ and satisfy $f_j = x^{j-1} + O(x^j)$ for $j=1,\ldots, n$. Since $g_i = x^{i-1} + O(x^i)$, it follows that $\frac{f_j}{g_i} = x^{j-i} + O(x^{j-i+1})$. Therefore for a fixed $i$, a linear combination of $\frac{f_j}{g_i}$ with $j=1,\ldots,n$ lies in $C[[x]]$ if and only if the linear combination involves only  $\frac{f_j}{g_i}$ for $j=i,\ldots, n$. 
		\item Since $T_k$ is a $C$-linear map, it follows from~\ref{it: wi-2} that $T_k(W_i) = \Span_C\left\{T_k\left(\frac{f_i}{g_i}\right), \ldots, T_k\left(\frac{f_n}{g_i}\right)\right\}$ for all~$k\geq0$. By Corollary~\ref{cor: series division}, we have
		$T_k\left(\frac{f_j}{g_i}\right) = T_k\left(\frac{T_{k+r-1}(f_j)}{T_{k+r-1}(g_i)}\right)$ for all $j=i,\ldots,n$ and all $k\geq 0$. Thus we obtained the desired result.
		\item Since $\frac{f_j}{g_i} = x^{j-i} + O(x^{j-i+1})$, it follows that $T_k\left(\frac{f_i}{g_i}\right), \ldots, T_k\left(\frac{f_n}{g_i}\right)$ are linearly independent over $C$ for all $k\geq n-1$ and all $i=1,\ldots,r$. Therefore, for all $k\geq n-1$,  $\dim_CT_{k}(W_i)=n-i+1$ and $T_k$ is an injective map from $T_{k+1}(W_i)$ to $T_{k}(W_i)$. Since $\bigcap_{i=1}^rT_{k+1}(W_i)$ is a subspace of $T_{k+1}(W_i)$, it follows that $T_k$ is an injective map from $\bigcap_{i=1}^rT_{k+1}(W_i)$ to $\bigcap_{i=1}^r T_{k}(W_i)$.\qedhere
	\end{enumerate}
\end{proof}

\begin{theorem}\label{thm: order computation}
	With Convention~\ref{con}, let $N=M\otimes L^{\otimes (r-1)} \in C(x)[D]$ and $k_0= \max \Z_N-\frac{r(r-1)}{2}$. Then for all $k>k_0$, 
	\begin{enumerate}[label=(\roman*)]
		\item $\dim_C(V(M): V(L)) = \dim_C\bigcap_{i=1}^r T_k(W_i)$;
		\item $T_k(V(M): V(L)) = \bigcap_{i=1}^r T_k(W_i)$.
	\end{enumerate}
\end{theorem}

\begin{proof}
	By Lemma~\ref{lem: w_i}, we obtain that $(V(M): V(L)) = \bigcap_{i=1}^r W_i$, where $W_i =\Span_C\left\{\frac{f_i}{g_i},\ldots,\frac{f_n}{g_i}\right\}$. To determine the dimension of the intersection of $W_i$, we shall multiply $W_i$ by $g:=g_1\ldots g_r$ and consider the solution space $V$ of $N=M\otimes L^{\otimes (r-1)}$ in $C[[x]]$. 
	
	For each $g_i$, let $\bar g_i = \frac{g}{g_i} = \prod_{j=1,j\neq i}^r g_i \in C[[x]]$. Then for each $1\leq i\leq r$,
	\[gW_i =\Span_C\left\{\bar g_if_i,\ldots, \bar g_i f_n\right\}\subseteq C[[x]],\]
	is a subspace $V$, because $\bar g_i$ is a solution of $L^{\otimes (r-1)}$, and $f_j$ is a solution of $M$ for $j=1,\ldots,n$. Note that $g=x^{\frac{r(r-1)}{2}} + O(x^{\frac{r(r-1)}{2}+1})\in C[[x]]$. It follows from Lemma~\ref{lem: w_i} and Proposition~\ref{prop: intersection dim} that for all $k>k_0$,
	
	\[\dim_C(V(M): V(L)) = \dim_{C}\left(\, \bigcap_{i=1}^rW_i \right) = \dim_C\left(\, \bigcap_{i=1}^r T_k(W_i) \right)\]
	and 
	\[T_k(V(M): V(L)) = T_k\left(\, \bigcap_{i=1}^r W_i \right) = \bigcap_{i=1}^r T_k(W_i).\qedhere\]
\end{proof}

\begin{example}\label{ex:2*3=4,order}
	We continue with Example~\ref{ex:2*3=4}. We want to compute the dimension of $(V(M): V(L))$ and a basis for it at precision $k>k_0$. We have $r=\ord(L)=2$, $n=\ord(M)=4$, and \begin{align*}
		N=M\otimes L^{\otimes (r-1)}=\ &(x^2 - 2x + 2)^2(x - 1)^5D^5+ 5(x^2 - 2x + 2)(x^2 - 2x - 2)(x - 1)^4D^4  \\&+ 40(x^2 - 2x + 4)(x - 1)^3D^3  - 120(x^2 - 2x + 6)(x - 1)^2D^2 \\&+120(x - 1)(2x^2 - 4x + 17)D -120(2x^2 - 4x + 23).
	\end{align*}
	The indicial polynomial of $N$ at $0$ is $\ind_0(N)=s(s - 1)(s - 2)(s - 3)(s - 4)$. The set of nonnegative roots of $\ind_0(L)$ is $\Z_N = \{0,1,2,3,4\}$. Thus $k_0=\max \Z_N -\frac{r(r-1)}{2}=3$. 
	
	By Lemma~\ref{lem: w_i}, we have
	$(V(M): V(L))=W_1\cap W_2$, where 
	
		$W_1=\Span_\bC\left\{\frac{f_1}{g_1},\frac{f_2}{g_1},\frac{f_3}{g_1},\frac{f_4}{g_1}\right\}=\Span_\bC\{1 - x + O(x^5),  x-x^2 +O(x^5), x^2- x^3 +O(x^5),x^3-x^4+O(x^5)\},$
		
		$W_2 =\Span_\bC\left\{\frac{f_2}{g_2},\frac{f_3}{g_2},\frac{f_4}{g_2}\right\}=\Span_\bC\{1 - x + O(x^5),  x-x^2 +O(x^5), x^2- x^3 +O(x^5)\}.$

	\noindent Let $k=k_0+1=4$. Then by Theorem~\ref{thm: order computation}, 
	\begin{equation}\label{eq: 2*3=4, basis}
		T_4(V(M): V(L))=T_4(W_1)\cap T_4(W_2)=\Span_\bC\{1 - x + O(x^5),  x-x^2 +O(x^5), x^2- x^3 +O(x^5)\}.
	\end{equation}
	Since the above truncated spaces have dimension three, Theorem~\ref{thm: order computation} implies that $(V(M): V(L))$ also has dimension three. The truncated basis in~\eqref{eq: 2*3=4, basis} can be uniquely extended to a basis of $(V(M): V(L))$.
\end{example}

\section{Symmetric products of generalized indicial polynomials}\label{sec:indpolys}
Given two linear differential operators $L,Q\in C(x)[D]$, the (generalized) local exponents of their symmetric product $L\otimes Q$ were studied by Singer~\cite[Lemma 3.2]{Singer93} and by van Hoeij and Weil~\cite[\S 3]{vanHoeijWeil97}. However, the multiplicities of these local exponents remain unclear. 
In terms of indicial polynomials, we will show that the symmetric product of their indicial polynomials is a divisor of the indicial polynomial of their symmetric product $L\otimes Q$. This result also extends to the generalized indicial polynomials. To avoid ambiguity, we use $s$ as the variable in the (generalized) indicial polynomials. Here, the symmetric product in $C[s]$ refers to the symmetric product of linear differential operators with constant coefficients.

\begin{lemma}\label{lem: indicial - multiplicity}
	Let $L\in C[x^{1/v}][D]$ for some $v\in \bN\setminus \{0\}$. If $L$ has a solution $g$ in $x^\alpha C[[x^{1/v}]][\log(x)]$ with initial term $x^\alpha \log(x)^{\mu-1}$, then $\alpha$ is a root of the indicial polynomial $\ind_0(L)$ with multiplicity at least $\mu$.
\end{lemma}
\begin{proof}
	Suppose that \[\ind_0(L) =  u (s)(s-\alpha_1)^{\mu_1}\cdots(s-\alpha_I)^{\mu_I} ,\] where $\mu_1,\ldots,\mu_I\in \bN\setminus\{0\}$, the roots $\alpha_1,\ldots, \alpha_I\in \alpha+\bZ$ are distinct, and $u(x)\in C[x]$ does not have any root in $\alpha+\bZ$. We may further assume that \[\alpha_1<\cdots<\alpha_I,\]
	where $\alpha_i<\alpha_{i'}$ means $\alpha_i - \alpha_{i'}<0$. Now we consider the ring of all finite $C$-linear combinations of series of the form $x^\beta b(x,\log(x))$ with $\beta\in \alpha+\bZ$ and $b\in C[[x^{1/v}]][y]$. By Lemma~\ref{lem: indicial - log}, the solution space of $L$ in this ring has a basis of the form:
	\[g_{i,j} = x^{\alpha_i}\log(x)^{j-1} + \cdots \quad ( i=1,\ldots, I, j = 1,\ldots, \mu_i),\]
	where $x^{\alpha_i} \log(x)^{j-1}$ is the initial term of $g_{i,j}$. Then $g$ is a linear combination of the $g_{i,j}$'s. 
	
	Suppose that
	\[g=\sum_{i=1}^I\sum_{j=1}^{\mu_i} c_{i,j}g_{i,j}\]
	for some nonzero $c_{i,j}\in C$. Let $i_0\in\{1,\ldots,I\}$ be the minimal index $i$ such that $g_{i,j}$ appears in this linear combination for some $j$, and $j_0\in\{1,\ldots,\mu_{i_0}\}$ be the minimal index $j$ such that $g_{i_0,j}$ appears. Then the initial term of $\sum_{i=1}^I\sum_{j=1}^{\mu_i} c_{i,j}g_{i,j}$ is $x^{\alpha_{i_0}} \log(x)^{j_0-1}$. By the assumption, the initial term of $g$ is $x^{\alpha}\log(x)^{\mu-1}$. Comparing initial terms, we conclude that $\alpha=\alpha_{i_0}$ and $\mu-1=j_0-1\leq \mu_{i_0}-1$, which implies $\mu\leq \mu_{i_0}$. Thus $\alpha$ is a root of $\ind_{0}(L)$ with multiplicity at least $\mu$.
\end{proof}

\begin{theorem}\label{thm: generalized indicial poly}
	Let $L,Q\in C(x)[D]$ and let $p,q\in C[x^{1/v}]$ with $v\in \bN\setminus\{0\}$. Let $\ind_{0,\exp(p(x^{-1}))}(L)$, $\ind_{0,\exp(q(x^{-1}))}(Q)$ be the generalized indicial polynomial of $L$ and $Q$ at $x=0$ with respect to exponential parts $\exp(p(x^{-1}))$ and $\exp(q(x^{-1}))$, respectively. Then the symmetric product \[\ind_{0,\exp(p(x^{-1}))}(L) \otimes \ind_{0,\exp(q(x^{-1}))}(Q)\] divides $\ind_{0,\exp(p(x^{-1})+q(x^{-1}))}(L\otimes Q)$.
\end{theorem}
\begin{proof}
	Suppose that $\alpha\in C$ is a root of $\ind_{0,\exp(p(x^{-1}))}(L)$ with multiplicity $\mu$, and $\beta\in C$ is a root of $\ind_{0,\exp(q(x^{-1}))}(Q)$ with multiplicity $\lambda$. By Lemma~\ref{lem: C-finite sym prod}, it suffices to prove that $\alpha+\beta$ is a root of $\ind_{0,\exp(p(x^{-1})+q(x^{-1}))}(L\otimes Q)$ with multiplicity at least $\mu+\lambda-1$.
	
	By Definition~\ref{Def: generalized indicial poly}, $\alpha$ is a root of the indicial polynomial of $\tilde L = \exp(-p(x^{-1}))\,L\, \exp(p(x^{-1}))$ with multiplicity $\mu$. By Lemma~\ref{lem: indicial - log}, $\tilde L$ has a solution $\tilde g\in x^\alpha C[[x]][\log(x)]$ with initial term $x^\alpha\log(x)^{\mu-1}$. Therefore $L$ has a solution \[g(x) = \exp(p(x^{-1}))\tilde g(x).\] Similarly, the operator $Q$ has a solution \[h(x) = \exp(q(x^{-1}))\tilde{h}(x),\] where $\tilde{h}\in x^\beta C[[x^{1/v}]][\log(x)]$ with initial term $x^\alpha\log(x)^{\lambda-1}$. By definition of symmetric product, 
	\[f(x) = g(x)h(x) = \exp(p(x^{-1})+q(x^{-1}))\tilde f(x)\]
	is a solution of $M:=L\otimes Q$, where
	\[\tilde f(x) = \tilde g(x)\tilde h(x) =  x^{\alpha+\beta}\log(x)^{\mu +\lambda-2}+\cdots \in x^{\alpha+\beta} C[[x]][\log(x)]\]
	with the initial term $x^{\alpha+\beta}\log(x)^{\mu +\lambda-2}$. Thus $\tilde f$ is a solution of 
	\[\tilde M = \exp(-p(x^{-1})-q(x^{-1})) \, M\, \exp(p(x^{-1})+q(x^{-1})) \in C[x^{1/v},x^{-1/v}][D].\]
	Applying Lemma~\ref{lem: indicial - multiplicity} to the operator $\tilde M$, we obtain that $\alpha+\beta$ is a root of the indicial polynomial $\ind_0(\tilde{M})$ with multiplicity at least $\lambda+\mu-1$. By Definition~\ref{Def: generalized indicial poly}, $\alpha+\beta$ is a root of the generalized indicial polynomial $\ind_{0,\exp(p(x^{-1})+q(x^{-1}))}({M})$ with multiplicity at least $\lambda+\mu-1$.
\end{proof}

Taking $p=q=0$ in Theorem~\ref{thm: generalized indicial poly} yields the following corollary.
\begin{corollary}\label{cor: indicial - sym prod}
	Let $L,Q\in C(x)[D]$ and let $\ind_0(L), \,\ind_0(Q)$ be their indicial polynomials at $0$. Then the symmetric product \[\ind_0(L) \otimes \ind_0(Q)\]  divides $\ind_0(L\otimes Q)$.
\end{corollary}
\begin{example}\label{ex:2*2=4,indicial}
	Let $L= (2x-1)D^2 -4xD+4$, $Q=(x-1)D^2 -xD +1\in \bC(x)[D]$. The point $\frac{1}{2}$ is an apparent singularity of $L$, but an ordinary point of $Q$ and $L\otimes Q$. We have
	\begin{align*}
		\ind_{\frac{1}{2}}(L)&=s(s-2),\\ 
		\ind_{\frac{1}{2}}(Q)&=s(s-1),\\
		\ind_{\frac{1}{2}}(L\otimes Q) &= s(s-1)(s-2)(s-3).
	\end{align*}
	Then $\ind_{\frac{1}{2}}(L)\otimes \ind_{\frac{1}{2}}(Q)=s(s-1)(s-2)(s-3)$ divides $\ind_{\frac{1}{2}}(L\otimes Q)$. The roots $\xi_{i}$ $(i=1,2,3,4)$  of $12x^4-44x^3+63x^2-52x+18$ are apparent singularities of $L\otimes Q$, but ordinary points of $L$ and $Q$. For each $\xi_i$, we have
	\begin{align*}
		\ind_{\xi_i}(L)&=s(s-1),\\
		\ind_{\xi_i}(Q)&=s(s-1), \\
		\ind_{\xi_i}(L\otimes Q)&=s(s-1)(s-2)(s-4).
	\end{align*} 
	Then $\ind_{\xi_i}(L)\otimes \ind_{\xi_i}(Q) = s(s-1)(s-2)$ divides $\ind_{\xi_i}(L\otimes Q)$.
\end{example}

Let $L,M\in C(x)[D]$. If $M=LQ$ for some $Q\in C(x)[D]$ under usual multiplication, then $\ind_0(Q)$ divides $\ind_0(M)$. So one can find the possible indicial polynomials of a right factor by factoring $\ind_0(M)$, see~\cite{BostanRivoalSalvy2024,vanHoeij97b}. For~symmetric product, if $M=L\otimes Q$ for some $Q\in C(x)[D]$, combining Corollary~\ref{cor: indicial - sym prod} and Proposition~\ref{prop: maximal quasi quotient} yields that the indicial polynomial $\ind_0(Q)$ divides the global quasi-symmetric quotient $\qsquo(\ind_0(M),\ind_0(L))$. So we can compute the possible indicial polynomials of a symmetric quotient. The procedure extends to generalized indicial polynomials as follows, see examples in the next section.
\begin{prop}\label{prop: quotient- exponential parts}
	Let $L,M\in C(x)[D]$ be of positive order and let $0\neq Q\in C(x)[D]$ be such that $L\otimes Q$ is a right factor of $M$. Then one can determine a finite set $\{\exp(q_i(x^{-1}))\}_{i=1}^\kappa$
	where $q_i\in C[x^{1/v}]$ with $v\in\bN\setminus\{0\}$, consisting of candidates for the exponential parts of the series solutions of $Q$ at $0$. Moreover, for each $\exp(q_i(x^{-1}))$, one can compute a polynomial $\widetilde{\ind}_{0,\exp(q_i(x^{-1}))}(Q)\in C[s]$ that is a multiple of the generalized indicial polynomial $\ind_{0,\exp(q_i(x^{-1}))}(Q)$.
\end{prop}
\begin{proof}
		 Let $\{\exp(p_j(x^{-1}))\}_{j=1}^\eta$ and $\{\exp(w_t(x^{-1}))\}_{t=1}^\rho$ be the exponential parts of the series solutions of $L$ and $M$ at $0$, respectively, where $p_j,w_\rho \in C[x^{1/v}]$ with $v\in\bN\setminus\{0\}$. If $\exp(q(x^{-1}))$ is an exponential part of $Q$ at $0$, then for all $1\leq j\leq \eta$, $\exp(q(x^{-1})+p_j(x^{-1}))$ is an exponential part of $M$ at $0$. Thus the exponential parts of $Q$ at $0$ belong to the set
		 \[\bigcap_{j=1}^\eta\left\{\exp(w_1(x^{-1})-p_j(x^{-1})),\ldots,\exp(w_\rho(x^{-1})-p_j(x^{-1}))\right \},\] 
		 where two exponential parts are considered identical if they differ by multiplication by a nonzero constant in $C$. Let $\{\exp(q_1(x^{-1})),\ldots,\exp(q_\kappa(x^{-1}))\}$ denote this set.
		 
		 For a fixed $\exp(q_i(x^{-1}))$, let 
		 \[\Lambda_i := \{(t,j)\mid \exp(q_i(x^{-1})) = c\exp(w_t(x^{-1})-p_j(x^{-1})) \text{ for some } c\in C\setminus\{0\}\}.\]
		 Then for each pair $(t,j)\in \Lambda_i$, by Theorem~\ref{thm: generalized indicial poly}, we obtain that \[\ind_{0,\exp(q_i(x^{-1}))}(L)\otimes \ind_{0,\exp(p_j(x^{-1}))}(Q) \,\mid\, \ind_{0,\exp(w_t(x^{-1}))}(L\otimes Q) \,\mid\, \ind_{0,\exp(w_t(x^{-1}))}(M).\]
		 Let $\mu_j(s):=\ind_{0,\exp(p_j(x^{-1}))}(L)$ and $\nu_\rho(s):= \ind_{0,\exp(w_t(x^{-1}))}(M)$. 
		  By Proposition~\ref{prop: maximal quasi quotient}, we get that $\ind_{0,\exp(q_i(x^{-1}))}(Q)$ divides the global quasi-symmetric quotient $\qsquo(\nu_t(s), \mu_j(s))$ for all $(t,j)\in \Lambda_i$. Thus we can take
		  \begin{equation}
		  	\widetilde{\ind}_{0,\exp(q_i(x^{-1}))}(Q):= \gcd_{(t,j)\in \Lambda_i}\qsquo(\nu_t(s),\mu_j(s)).\qedhere
		  \end{equation}
\end{proof}

\section{Degree bounds for symmetric quotients}\label{sec:degree}
Let $L,M\in C(x)[D]$ be given. Let $Q=D^\delta + b_{\delta-1}(x)D^{\delta-1}+ \cdots + b_0(x)$ with $b_i\in C(x)$ be such that $L\otimes Q$ is a right factor of $M$. In this section, we compute bounds for the degrees of the numerators and denominators of the $b_i$. Our work is inspired by the computation of degree bounds for a right factor of a given linear differential operator~\cite{vanHoeij97a,BostanRivoalSalvy2021explicit,BostanRivoalSalvy2024}; see~\cite{BostanRivoalSalvy2021explicit} for a detailed computation and~\cite{BostanBKM2016} for an explicit and challenging example. Similarly, these bounds for $Q$ are known when:
\begin{itemize}
	\item for every $b_i$ and for every point $\xi\in\Sing(L)\cup\Sing(M)\cup\{\infty\}$, we have a lower bound for the valuation of $b_i \in C(x)$ at $\xi$;
	\item we have an upper bound for the number of extra singularities. A point $\xi\in C$ is called an \emph{extra singularity} of the quotient $Q$ if $\xi$ is an ordinary point of $L$ and $M$, but a singularity of $Q$.
\end{itemize}

By Proposition~\ref{prop: quotient- exponential parts}, we can compute the possible exponential parts of $Q$ at $0$. Let $\exp(q(x^{-1}))$ with $q\in \bigcup_{v\in \bN\setminus\{0\}}C[x^{1/v}]$ be one of them such that $c:=-\deg(q)$ is minimal. Then $1-c$ is the largest possible slope of Newton Polygon of $Q$ at $0$, see~\cite[\S 3.4]{kauers23}. A lower bound for the valuation of $b_i$ at $0$ can be obtained from the study of the Newton Polygon of $Q$ at $0$, see~\cite{BostanRivoalSalvy2021explicit}. The same process can be performed at every point $\xi\in C\cup\{\infty\}$. So we only need an upper bound for the number of extra singularities. 
\subsection{The Fuchsian case}
Assume that $Q$ is Fuchsian. Note that $L,M$ need not necessarily be Fuchsian. Let $\Extra(Q)$ be the set of all extra singularities of $Q$. If $\xi\in C$ is an extra singularity of $Q$, then by Proposition~\ref{prop:order}, $\xi$ is an apparent singularity of $Q$. Therefore the quantity $S_\xi(Q)$ in~\eqref{EQ: S_xi(L)-fuchsian} is a positive integer. So applying the Fuchs relation~\eqref{EQ: fuchs rel} to $Q$, the number of extra singularities is upper bounded by
\begin{equation}\label{EQ: extra_Q - fuchs}
	\# \Extra(Q) \leq \sum_{\xi\in \Extra(Q)} S_\xi(Q) = -\delta(\delta-1) - \sum_{\xi\in \Sing^*(Q)\cup\{\infty\}}S_\xi(Q),
\end{equation}
where $\Sing^*(Q)$ is a subset of $\Sing(Q)$ that are not extra singularities of $Q$. By the definition of extra singularities, we get $\Sing^*(Q)\subseteq \Sing(L)\cup\Sing(M)$. By Proposition~\ref{prop: quotient- exponential parts}, for each $\xi\in C\cup\{\infty\}$, one can compute a multiple of $\ind_\xi(Q)$, denoted by $\widetilde{\ind}_\xi(Q)$. For example, the quasi-symmetric quotient $\qsquo(\ind_\xi(M),\ind_\xi(L))$ is a multiple of $\ind_\xi(Q)$. Then the roots of $\ind_\xi(Q)$ are roots of $\widetilde{\ind}_{\xi}(Q)$. Therefore, by~\eqref{EQ: S_xi(L)-fuchsian}, we get $S_\xi(Q)\geq \widetilde{S}_\xi(Q)$, where $\widetilde S_\xi(Q)$ denotes the sum of the $\delta$ smallest roots of $\widetilde{\ind}_\xi(Q)$, minus~$\frac{\delta(\delta-1)}{2}$. It follows from \eqref{EQ: extra_Q - fuchs} that
\[\# \Extra(Q) \leq -\delta(\delta-1) - \sum_{\xi\in\Sing(L)\cup\Sing(M)\cup\{\infty\}} \widetilde{S}_\xi(Q).\]

This process can be used whenever the operator $Q$ to be found is known to be Fuchsian. In particular, when $L$ and $M$ are Fuchsian, Proposition~\ref{prop: quotient- exponential parts} implies that $Q$ is Fuchsian.

If the degree of $\widetilde{\ind}_\xi(Q)$ is less than $\delta$, i.e., the number of roots of $\widetilde{\ind}_\xi(Q)$ in $C$ is less than $\delta$, then there is no operator $Q$ of order $\delta$ such that $L\otimes Q$ is a right factor of $M$.

\begin{example}\label{ex:2*3=4,degree}
	We continue with Examples \ref{ex:2*3=4} and~\ref{ex:2*3=4,order}. We show how to compute a degree bound for an unknown operator $Q\in \bC(x)[D]$ of order $3$ such that $L\otimes Q$ is a right factor of $M$. Both $L$ and $M$ have four singularities: $1,\xi_1, \xi_2$ and $\infty$, where $\xi_1,\xi_2$ are distinct roots of $x^2-2x+2$. These four singularities are regular. Hence $L$ and $M$ are Fuchsian, and therefore $Q$ is also Fuchsian. At the point $1$, we have 
	\[\ind_1(M)=(s-2)(s-3)(s-4)(s-5),\quad \ind_1(L)=(s-1)(s-2).\]
	Then $\widetilde{\ind}_1(Q) =\qsquo(\ind_1(M),\ind_1(L)) = (s-1)(s-2)(s-3)$ is a multiple of $\ind_1(Q)$. At the point $\xi_i$ $(i=1,2)$, we have 
	\[\ind_{\xi_i}(M)=(s+1)s(s-1)(s-2),\quad \ind_{\xi_i}(L)=(s+1)s.\]
	Then $\widetilde{\ind}_{\xi_i}(Q) =\qsquo(\ind_{\xi_i}(M),\ind_{\xi_i}(L)) = s(s-1)(s-2)$ is a multiple of $\ind_{\xi_i}(Q)$. At the point $\infty$, we have 
	\[\ind_{\xi_i}(M)=(s+3)(s+2)(s+1)s,\quad \ind_{\xi_i}(L)=s(s-1).\]
	Then $\widetilde{\ind}_{\infty}(Q) =\qsquo(\ind_{\infty}(M),\ind_{\infty}(L)) = (s+3)(s+2)(s+1)$ is a multiple of $\ind_{\infty}(Q)$.
	Therefore, for the operator $Q$,
	\begin{align*}
		\widetilde{S}_1(Q) &= 1+2+3-3=3,\\
		\widetilde{S}_{\xi_i}(Q) &= 0+1+2-3=0,\\
		\widetilde{S}_{\infty}(Q) &= -3-2-1-3=-9.
	\end{align*}
	Thus the number of extra singularities of $Q$ is upper bounded by:
	\[\# \Extra(Q)\leq -3(3-1) - (3+0+0-9) = 0.\]
	This implies that $Q$ has no extra singularities. 
	
	Since $Q$ is Fuchsian, it can be written
	\[Q = D^3 + \frac{a_2(x)}{A(x)}D^2 + \frac{a_1(x)}{A(x)^2} D + \frac{a_0(x)}{A(x)^3},\]
	where $a_i, A\in \bC[x]$ and $\deg(a_i)\leq \deg(A^i) - (3-i)$. Suppose $A(x)=A_1(x)A_2(x)$, where the roots of $A_1$ are elements of $\Sing^*(Q)$ and the roots of $A_2$ are elements of $\Extra(Q)$. It follows that 
	\[\deg(A_1)\leq \# \Sing^*(Q)\leq \# (\Sing(L)\cup\Sing(M)) = \#\{1,\xi_1,\xi_2\}=3\]
	and $\deg(A_2)\leq \# \Extra(Q)=0$. Clearing the denominator of $Q$ gives the bounds $(27,26,25,24)$ on the degrees of the coefficients of $(D^3,D^2,D,1)$. This is the bound used in Example~\ref{ex:2*3=4} leading to the discovery of the symmetric quotient $Q$. From Example~\ref{ex:2*3=4,Q}, we see that $Q$ has only two singularities $1$ and $\infty$. Thus, $Q$ indeed has no extra singularities.
%
\end{example}
\subsection{The general case}
Applying the generalized Fuchs relation~\eqref{EQ: fuchs rel generalized} to $Q$, we obtain the analogue of~\eqref{EQ: extra_Q - fuchs}:
\begin{equation}\label{EQ: extra_Q}
	\# \Extra(Q) \leq \sum_{\xi\in \Extra(Q)} \left(S_\xi(Q) -\frac{1}{2}I_\xi(Q)\right)= -\delta(\delta-1) - \sum_{\xi\in \Sing^*(Q)\cup\{\infty\}}\left(S_\xi(Q)-\frac{1}{2}I_\xi(Q)\right).
\end{equation}
As in the Fuchsian case, $\Sing^*(Q)\subseteq \Sing(L)\cup \Sing(M)$.
By Proposition~\ref{prop: quotient- exponential parts}, for each $\xi\in C$, one can compute the possible exponential parts $\{\exp(q_i((x-\xi)^{-1}))\}_{i=1}^{\kappa}$ of the series solutions of $Q$ at $\xi$, where $q_i\in C[x^{1/v}]$ with $v\in \bN\setminus\{0\}$ and $q_i(0)=0$. One can also compute a multiple of the generalized indicial polynomial $\ind_{\xi,\exp(q_i((x-\xi)^{-1}))}(Q)$, denoted by $\widetilde{\ind}_{\xi,\exp(q_i((x-\xi)^{-1}))}(Q)$. Therefore, by~\eqref{EQ: S_xi(L)}, we get $S_\xi(Q)\geq \widetilde{S}_\xi(Q)$, where $\widetilde S_\xi(Q)$ denotes the sum of the $\delta$ smallest roots of $\prod_{i=1}^\kappa\widetilde{\ind}_{\xi,\exp(q_i((x-\xi)^{-1}))}(Q)$, minus~$\frac{\delta(\delta-1)}{2}$. 

For each $1\leq i\leq \kappa$, there are at most $d_i$ linearly independent solutions of $Q$ at $\xi$ with the exponential part $\exp(q_i((x-\xi)^{-1}))$, where $d_i$ is the degree of $\widetilde{\ind}_{\xi,\exp(q_i((x-\xi)^{-1}))}(Q)$. So counting  $\exp(q_i((x-\xi)^{-1}))$ repeated $d_i$ times, we get a list $\exp(\tilde q_1((x-\xi)^{-1})), \ldots, \exp(\tilde q_{\tilde \delta}((x-\xi)^{-1}))$ of the possible exponential parts for the operator $Q$ at $\xi$, where $\tilde \delta = \sum_{i=1}^{\kappa} d_i$. Therefore, by~\eqref{EQ: I_xi(L)}, we get $I_\xi(Q)\geq \widetilde{I}_\xi(Q)$, where $\widetilde{I}_\xi(Q)$ denotes twice the sum of the $\frac{1}{2}\delta(\delta-1)$ smallest values among $\deg(\tilde q_i - \tilde q_j)$ for all $1\leq i< j\leq \tilde \delta$. 

The case at $\xi=\infty$ is similar. It then follows from \eqref{EQ: extra_Q} that
\[\# \Extra(Q) \leq -\delta(\delta-1) - \sum_{\xi\in\Sing(L)\cup\Sing(M)\cup\{\infty\}} \left(\widetilde{S}_\xi(Q)-\frac{1}{2}\widetilde{I}_\xi(Q)\right).\]

If $\tilde \delta<\delta$, then there is no operator $Q$ of order $\delta$ such that $L\otimes Q$ is a right factor of $M$.

To compute a sharper degree bound for $Q$, one may use integer linear programming as in~\cite{BostanRivoalSalvy2024}.

\begin{example}
	Let $L,M\in \bC(x)[D]$ be two operators:
	\begin{align*}
			L = & \ (2x-1)D^2 -4xD+4,\\
			M= & \ (12x^4 - 44x^3 + 63x^2 - 52x + 18) D^4 +(-72x^4 + 216x^3 - 246x^2 + 186x - 56)D^3 \\
			&+(132x^3 - 232x^2 + 99x - 96)xD^2 + (-72x^4 - 108x^3 + 232x^2 + 18x + 96)D +  144x^3 - 48x^2 - 234.
		\end{align*}
	Since $\dim_\bC(V(M):V(L))=2$, we assume that $Q\in \bC(x)[D]$ is an operator of order two such that $L\otimes Q$ is a right factor of $M$. We compute an upper bound on the number of extra singularities of $Q$. The singularities of $L$ are: $\frac{1}{2}$ and $\infty$. The singularities of $M$ are: $\xi_i$ $(i=1,2,3,4)$ and $\infty$, where the $\xi_i$ are distinct roots of $12x^4 - 44x^3 + 63x^2 - 52x + 18$. As shown in Example~\ref{ex:2*2=4,indicial}, the point $\frac{1}{2}$ is an apparent singularity of $L$. The points $\xi_i$ are apparent singularities of $M$. Similar to Example~\ref{ex:2*3=4,degree} in the Fuchsian case, we have
	\begin{align*}
		&\widetilde{\ind}_{\frac{1}{2}}(Q) =\qsquo(\ind_{\frac{1}{2}}(M),\ind_{\frac{1}{2}}(L)) =\qsquo(s(s-1)(s-2)(s-3), s(s-2))= s(s-1),\\
		&\widetilde{\ind}_{\xi_i}(Q) =\qsquo(\ind_{\xi_i}(M),\ind_{\xi_i}(L)) =\qsquo(s(s-1)(s-2)(s-4), s(s-1))= s(s-1).
	\end{align*}
	The point $\infty$ is an irregular singularity of $L$ and $M$. The generalized indicial polynomials of $L$ are \[\ind_{\infty, \exp(0)}(L) = 2(s+1),\quad\ind_{\infty, \exp(2x)}(L) = -2s.\] 
	The generalized indicial polynomials of $M$ are \[\ind_{\infty, \exp(0)}(M){=}6(s+2), \,\ind_{\infty, \exp(x)}(M){=} {-}2(s+1), \,\ind_{\infty, \exp(2x)}(M) {=} 2(s+1), \,\ind_{\infty, \exp(3x)}(M) {=} {-}6s.\] 
	Thus by Proposition~\ref{prop: quotient- exponential parts}, the possible exponential parts of $Q$ are
	\[\{\exp(0),\exp(x),\exp(2x),\exp(3x)\}\cap\{\exp(-2x),\exp(-x),\exp(0),\exp(x)\}=\{\exp(0),\exp(x)\}.\]
	Since $\exp(0) = \exp(0-0) =\exp(2x-2x)$, we have 
	\[\widetilde{\ind}_{\infty,\exp(0)}(Q) = \gcd(\qsquo(6(s+2), 2(s+1)), \, \qsquo(2(s+1), -2s)) =\gcd(s+1,s+1)=s+1.\]
	Since $\exp(x) = \exp(x-0) =\exp(3x-2x)$, we have 
	\[\widetilde{\ind}_{\infty,\exp(x)}(Q) = \gcd(\qsquo(-2(s+1), 2(s+1)), \, \qsquo(-6s, -2s)) =\gcd(s,s)=s.\]
	Thus, for the operator $Q$,
	\begin{align*}
		\widetilde S_\frac{1}{2}(Q) &= 0+1- 1 = 0,\\
		\widetilde S_{\xi_i}(Q) &= 0+1- 1 = 0,\\
		\widetilde S_\infty(Q) &= 0-1- 1 = -2,\\
		\widetilde I_\infty(Q) &= 2\cdot 1= 2.
	\end{align*}
	It follows that
	\[\# \Extra(Q) \leq -2(2-1) - (0+0+0+0+0-2-\frac{1}{2}2) =1.\]
	Since $M=L\otimes Q$, with $Q$ as in Examples \ref{ex:2*2=4,Q} and~\ref{ex:2*2=4,indicial}, we see that $Q$ has only two singularities: $1$ and and $\infty$. Thus $Q$ indeed has one extra singularity $1$. 
\end{example}
\section{Another example}
Algorithms for factoring operators of orders three and four with respect to symmetric product are known~\cite{Singer1985,Hessinger2001,vanHoeij2002,vanHoeij2007}. Here we give an example of computing a symmetric quotient of an operator of order nine by a factor of order three. Although this example does not fall into any of the three special cases described in Section~\ref{sec:special}, Algorithm~\ref{alg:global} successfully produces a symmetric quotient. 

Let $L=(x-1)^3D^3 + (5(x-1)^2 +(x-1)^3)D^2 + ((x-1)^2 - 17(x-1))D + 24\in \bC(x)[D]$ and let $M=L\otimes P$, where $P=(x-1)^3D^3+(11(x-1)^2+(x-1)^3)D^2 + 30(x-1)+18\in \bC(x)[D]$. Then $r=\ord(L)=3$, $n=\ord(M)=9$. The leading coefficient of $M$ is $m(x)(x-1)^8$, where
\begin{align*}
	m=&\ 1568x^{11} -161032x^{10} - 2017870x^9 + 31228120x^8 - 506595359x^7 + 8763692179x^6 - 91370341057x^5 \\
	&+ 610286581763x^4 - 3583448187232x^3 + 15654415322868x^2 - 37066249396506x + 24398082715566
\end{align*}
is an irreducible polynomial over $\bQ$ of degree $11$. Assume that $P$ is unknown. The goal is to compute the global quasi-symmetric quotient $Q$ of $M$ by $L$. 

 First we compute an upper bound for the order of $Q$. Since the indicial polynomial of $N=M\otimes L^{\otimes 2}$ is $\ind_0(N) =\prod_{i=0}^{29}(x-i)$, we take $k_0=29 - 3 = 26$ and $k=k_0+1=27$. The space $(V(M):V(L))$ has dimension three, with a basis given by
\begin{align*}
	&h_1 = 1 + 3x^3 + \frac{39}{4}x^4+\frac{201}{10}x^5+\frac{1343}{40}x^6+\frac{1737}{35}x^7+\frac{151717}{2240}x^8+\frac{125849}{1440}x^9+\frac{7269929}{67200}x^{10}+O(x^{11}),\\
	&h_2 = x-5x^3-\frac{57}{4}x^4-\frac{271}{10}x^5-\frac{5123}{120}x^6-\frac{301}{5}x^7-\frac{530161}{6720}x^8-\frac{593527}{6048}x^9-\frac{23664101}{201600}x^{10}+O(x^{11}),\\
	&h_3=x^2+\frac{10}{3}x^3+\frac{27}{4}x^4+\frac{65}{6}x^5+\frac{5471}{360}x^6+\frac{410}{21}x^7+\frac{94985}{4032}x^8+\frac{351137}{12960}x^9+\frac{18123281}{604800}x^{10}+O(x^{11}).
\end{align*}
So the order of $Q$ is at most three. Here if $k=9$, the truncated space $\bigcap_{i=1}^3T_k(W_i)$ has dimension four. When $k\geq10$, this dimension remains stable: $\dim_C(\bigcap_{i=1}^3T_k(W_i))=\dim_C(V(M):V(L))=3$. 

Now we compute the number of extra singularities of $Q$. The singularities of $L$ are: $1$ and $\infty$. The singularities of $M$ are: $1$, $\xi_i$ $(i=1,\ldots,11)$ and $\infty$, where the $\xi_i$ are distinct roots of $m(x)$. The point $1$ is a regular singularity of $L$ and $M$. The points $\xi_i$ are apparent singularities of $M$. So we have
\begin{align*}
	&\widetilde{\ind}_{1}(Q) =\qsquo(\ind_{1}(M),\ind_{1}(L)) =\qsquo(s^3(s+8)(s+9)^2(s+1)^3, (s+6)(s-2)^2)= (s+2)(s+3)^2\!,\\
	&\widetilde{\ind}_{\xi_i}(Q) =\qsquo(\ind_{\xi_i}(M),\ind_{\xi_i}(L)) =\qsquo((s-9)\textstyle\prod_{i=0}^7(s-i), s(s-1)(s-2))= \textstyle\prod_{i=0}^5(s-i).
\end{align*}
	The point $\infty$ is an irregular singularity of $L$ and $M$. The generalized indicial polynomials of $L$ are \[\ind_{\infty, \exp(0)}(L) =-s^2,\quad\ind_{\infty, \exp(-x)}(L) = s-4.\] 
The generalized indicial polynomials of $M$ are \[\ind_{\infty, \exp(0)}(M) = -2(s+1)^2s^2,\, \ind_{\infty, \exp(-x)}(M) = -(s-3)(s-4)(s-11)^2,\, \ind_{\infty, \exp(-2x)}(M) = 16(s-15).\] 
Thus by Proposition~\ref{prop: quotient- exponential parts}, the possible exponential parts of $Q$ are
\[\{\exp(0),\exp(-x),\exp(-2x)\}\cap\{\exp(x),\exp(0),\exp(-x)\}=\{\exp(0),\exp(-x)\}.\]
Since $\exp(0) = \exp(0-0) =\exp(-x-(-x))$, we have 
\[\widetilde{\ind}_{\infty,\exp(0)}(Q) = \gcd(\qsquo(-2(s+1)^2s^2, -s^2), \, \qsquo(-(s-3)(s-4)(s-11)^2, s-4)) =s(s+1).\]
Since $\exp(-x) = \exp(-x-0) =\exp(-2x-(-x))$, we have 
\[\widetilde{\ind}_{\infty,\exp(-x)}(Q) = \gcd(\qsquo(-(s-3)(s-4)(s-11)^2, -s^2), \, \qsquo(16(s-15), s-4)) =s-11.\]
Thus, for the operator $Q$,
\begin{align*}
	\widetilde S_1(Q) &= -2-3-3- 3 = -11,\\
	\widetilde S_{\xi_i}(Q) &= 0+1+2- 3 = 0,\\
	\widetilde S_\infty(Q) &= 0-1+11-3 = 7,\\
	\widetilde I_\infty(Q) &= 2\cdot 2= 4.
\end{align*}
It follows that
\[\# \Extra(Q) \leq -3(3-1) - (-11+0+7-\frac{1}{2}4) =0.\]
Therefore, $Q$ has no extra singularities and at most $13$ singularities: $1$, $\xi_i (i=1,\ldots, 11)$ and $\infty$. 

Since the singularities $1$, $\xi_i (i=1,\ldots, 11)$ are regular, the Newton polygons of $Q$ at each of these points have only one edge with slope $1$. At the point $\infty$, since the possible exponential parts of $Q$ are $\exp(0)$ and $\exp(-x)$, the possible slopes of the Newton polygon of $Q$ are $1$ and $2$. 
We write \[Q = D^3 + \frac{A_2(x)}{B(x)}D^2 + \frac{A_1(x)}{B(x)} D+\frac{A_0(x)}{B(x)},\] where $A_i,B\in C[x]$. Then $\deg(B)\leq 3(11+1+2)=45$ and $\deg(A_i)\leq \deg(B)+3(2-1)=48$, see details in~\cite{BostanRivoalSalvy2021explicit}. Clearing the denominator of $Q$ gives the bounds $(45,48,48,48)$ on the degrees of the coefficients of $(D^3,D^2,D,1)$. By solving the linear system $Q\cdot h_j=O(x^k)$ for $j=1,2,3$ and sufficiently large $k$, we find that $Q=P$ is the global quasi-symmetric quotient of $M$ by $L$. 

\bibliographystyle{plain}
\bibliography{bib}

\end{document}